\documentclass[aps, 10pt, prl, letterpaper, twocolumn, superscriptaddress]{revtex4-1}

\usepackage{graphicx, xcolor, times, amsmath, amsfonts, amssymb, amsthm, bbm, braket, comment}
\usepackage{hyperref}

\definecolor{darkblue}{RGB}{0,0,127} 
\definecolor{darkgreen}{RGB}{0,150,0}
\hypersetup{
	breaklinks, 
	colorlinks, 
	linkcolor=darkblue,
	citecolor=darkgreen, 
	urlcolor=blue,
	pdftitle={Comparing Experiments to the Fault-Tolerance Threshold}, 
	pdfauthor={Richard Kueng, David M. Long, Andrew C. Doherty, and Steven T. Flammia}
}

\newcommand{\tr}{{\mathrm{Tr}}}
\def\avg#1{\mathinner{\langle{#1}\rangle}}

\def\ket#1{\mathinner{|{#1}\rangle}}

\newcommand{\dnorm}[1]{\| #1 \|_\diamond}

\newcommand{\E}{\mathcal{E}}
\newcommand{\I}{\mathcal{I}}

\newcommand{\half}{\frac{1}{2}}

\newtheorem{corollary}{Corollary}
\newtheorem{theorem}{Theorem}

\newtheorem{proposition}{Proposition}

\begin{document}
\title{Comparing Experiments to the Fault-Tolerance Threshold}

\author{Richard Kueng}
\affiliation{Centre for Engineered Quantum Systems, School of Physics, University of Sydney, Sydney, NSW, Australia}
\affiliation{Institute for Theoretical Physics, University of Cologne, Germany}
\affiliation{Institute for Physics \& FDM, University of Freiburg, Germany}

\author{David M.\ Long}
\author{Andrew C.\ Doherty}
\author{Steven T.\ Flammia}
\affiliation{Centre for Engineered Quantum Systems, School of Physics, University of Sydney, Sydney, NSW, Australia}

\date{\today}

\begin{abstract}
Achieving error rates that meet or exceed the fault-tolerance threshold is a central goal for quantum computing experiments, and measuring these error rates using randomized benchmarking is now routine. However, direct comparison between measured error rates and thresholds is complicated by the fact that benchmarking estimates average error rates while thresholds
reflect worst-case behavior when a gate is used as part of a large computation.
These two measures of error can differ by orders of magnitude in the regime of interest. Here we facilitate comparison between the experimentally accessible average error rates and the worst-case quantities that arise in current threshold theorems by deriving relations between the two for a variety of physical noise sources.
Our results indicate that it is coherent errors that lead to an enormous mismatch between average and worst case, and we quantify how well these errors must be controlled to ensure fair comparison between average error probabilities and fault-tolerance thresholds. 
\end{abstract}

\maketitle

The fault-tolerance threshold theorem is a fundamental result that justifies the tremendous interest in building large-scale quantum computers despite the formidable practical difficulties imposed by noise and imperfections. This theorem gives a theoretical guarantee that quantum computers can be built in principle if the noise strength and correlation are below some threshold value~\cite{Kitaev1997, Aharonov1997, Knill1998b}. 

To make precise statements of threshold theorems, we must quantify the strength of errors in noisy quantum operations. Ideally we would do this in terms of quantities that can be measured in experiments. A standard measure for quantifying errors in quantum gates is given by the \emph{average error rate}, which is defined as the infidelity between the output of an ideal unitary gate $\mathcal{U}$ and a noisy version $\mathcal{E}\mathcal{U}$ with noise process $\mathcal{E}$, uniformly averaged over all pure states,
\begin{equation}
\label{eq:avgerrorrate}
	r(\mathcal{E}) = 1- \int \mathrm{d}\psi\, \langle\psi| \mathcal{E}\bigl(|\psi\rangle\!\langle\psi|\bigr) |\psi\rangle\,.
\end{equation}
This quantity has many virtues: it can be estimated efficiently for any ideal gate $\mathcal{U}$, and in a manner that is independent of state preparation and measurement (SPAM) errors by using the now standard method of randomized benchmarking~\cite{Emerson2005, Knill2008, Magesan2012b, Granade2014}. Recent experimental implementations include~\cite{chow2009,brown2011single,gaebler2012randomized,chow2012,corcoles2013,chow2014implementing,Barends2014,harty2014high,muhonen2015,xia2015}.

The major drawback of using Eq.~(\ref{eq:avgerrorrate}) to quantify gate errors is that it is only a proxy for the actual quantity of interest, the fault-tolerance threshold. This is because $r$ captures average-case behavior for a single use of the gate, while fault tolerance theorems characterize noise in terms of \emph{worst-case} performance when the gate is used repeatedly in a large computation. The importance of this distinction has recently been emphasized by Sanders {\it et al}~\cite{Sanders2015}. For some noise types (such as pure dephasing and depolarizing noise) the worst- and average-case behavior essentially coincide~\cite{Magesan2012}. However for other classes of errors, notably errors in detuning and calibration that lead to over or under rotation, the worst-case behavior is proportional to $\sqrt{r}$ and can be \emph{orders of magnitude} worse than the average in the relevant regime of $r \ll 1$, as we will discuss in more detail below. Thus it is not possible to directly compare a measured value of $r$ to a threshold result. Despite this, experimentalists are increasingly wishing to relate the results of benchmarking experiments to fault tolerance thresholds. There is thus a pressing need for techniques that allow for direct comparison between experimentally measurable error rates and fault-tolerance thresholds. 

In this Letter, we investigate the relationship between worst-case and average-case error for a wide range of error models that are relevant to experiments. Firstly, we show that while closed  form expressions do not typically exist, well-established theoretical techniques of convex optimization are often sufficient to determine the relationship between average-case and worst-case errors for models of physical interest. The details of these computations are largely relegated to the Supplementary Material. Secondly, we study a wide range of error models for one-qubit gates. Our main example is of a one-qubit gate with combined dephasing and calibration error. This allows us to demonstrate the crossover between a regime dominated by dephasing, where average-case and worst-case errors are not too different, and the limit of a unitary noise, where the worst-case error scales like $\sqrt{r}$. We then turn to general bounds on worst-case error, showing that it scales as $\sqrt{r}$ for all unitary errors and that for a wide class of errors it can be accurately estimated in terms of $r$ and a recently introduced measure of how close an error process is to being unitary. Finally, conventional benchmarking experiments contain a lot more information than is required just to extract $r$. We find that this information can often be used to show that the worst-case error has an unfavourable scaling. This is an area that we hope will attract much more study in future.

\paragraph{Fault-tolerance thresholds.}

A wide range of fault-tolerance thresholds have been reported. The value of the threshold depends greatly on the fault tolerant procedures that are used, on the noise model that is assumed, and whether the threshold is determined from (possibly conservative) analytic bounds on the error, or from (possibly optimistic) numerical simulations. We emphasize that the errors that are given in theoretical fault tolerance papers typically refer to some measure of worst-case error. For example the widely known results of Aliferis and collaborators~\cite{Aliferis2006a,Aliferis2007,Aliferis2009} use concatenated error correcting codes and consider a stochastic adversarial noise model that includes all of the noise processes that we will discuss in this paper. These papers find that large-scale quantum computation can be performed for errors below a few times $10^{-4}$, when that error is quantified by a measure of worst-case error such as the diamond distance that we discuss below. For more optimistic noise models and for fault-tolerant protocols such as the widely known surface code approaches, the threshold is around $10^{-2}$ based on numerical simulations of Pauli errors~\cite{Wang2011}. For Pauli noise however there is no significant difference between worst-case and average-case errors~\cite{Magesan2012}. The performance of these schemes in the presence of coherent errors is not yet understood.

It is possible to state a version of the threshold theorem directly in terms of $r$, but given current knowledge the thresholds in these theorems would be roughly the square of current thresholds (around $10^{-8}$ for~\cite{Aliferis2006a,Aliferis2007,Aliferis2009}). It is unclear if this can be significantly improved upon since it may be that it is the worst-case error that is physically relevant to the success of the computation. However, our results here motivate research into whether current fault tolerance results could be strengthened to provide significantly improved thresholds when expressed in terms of $r$ for error models sufficiently general to include coherent errors.

\paragraph{Diamond distance.}

We will now describe the most commonly used metric of worst-case error for quantum processes. Any candidate measure of distance $\Delta(\mathcal{E},\mathcal{F})$ between noise operations $\mathcal{E}$ and $\mathcal{F}$ should satisfy certain desirable properties~\cite{Gilchrist2005}. (The operation $\mathcal{F}$ should be thought of as a perfect identity gate for our purposes.) First, like any good distance measure it should have the structure of a metric, which in particular means it should be symmetric, positive, and obey the triangle inequality. Less obviously, but even more importantly, it should obey two additional properties: \emph{chaining} and \emph{stability}. The chaining property,
\begin{equation}
\label{eq:chaning}
	\Delta(\mathcal{E}_2\mathcal{E}_1,\mathcal{F}_2\mathcal{F}_1) \le \Delta(\mathcal{E}_1,\mathcal{F}_1)+\Delta(\mathcal{E}_2,\mathcal{F}_2)\,,
\end{equation}
says that composing two noisy operations cannot amplify the error by more than the sum of the two individual errors. Thus, errors can grow at most linearly in the number of operations. The stability property states that the error metric for a single gate should be independent of whether that gate is embedded in a larger computation. So we require
\begin{equation}
\label{eq:stability}
	\Delta(\mathcal{I}\otimes\mathcal{E},\mathcal{I}\otimes\mathcal{F}) = \Delta(\mathcal{E},\mathcal{F})\,,
\end{equation}
where $\mathcal{I}$ is the identity operation. This ensures that our measure is robust even if the input to the gate is entangled with other qubits in the computation. 

The diamond distance, whose formal definition is
\begin{equation}
\label{eq:diamond}
	D(\mathcal{E},\mathcal{F}) = \tfrac{1}{2} \max_{\rho} \|\mathcal{I}\otimes\mathcal{F}(\rho) - \mathcal{I}\otimes\mathcal{E}(\rho)\|_1\,,
\end{equation}
satisfies each of these physically motivated desiderata~\cite{Kitaev1997}. It also has an appealing operational interpretation as the maximum probability of distinguishing the output of the noisy gate from the ideal output~\cite{Kitaev1997, Helstrom1967}. It is not obvious from the definition how to do practical computations with this quantity, but it can be computed efficiently using the methods of semidefinite programming~\cite{Watrous2009, Ben-Aroya2010, Watrous2013}. Because of these properties, the diamond distance is an ideal measure for quantifying noise for the purposes of a fault-tolerance threshold, although in principle other quantities could be employed as well~\cite{Aharonov1997}. 

The only drawback of this quantity is that it is not known how to measure it directly in experiments. It is therefore of interest to have a conversion to, or at least bounds for, diamond distance in terms of the average gate fidelity. To date, the best known bounds for a $d$-level quantum gate are~\cite{Wallman2014}
\begin{equation*}
	\tfrac{d+1}{d}r \le D \le \sqrt{d (d+1) r}\,,
\end{equation*}
but it is unknown for what conditions these bounds are tight.  

\paragraph{Single-qubit calibration and dephasing errors.}

In order to discuss the relationship between average-case and worst-case errors in quantum computing demonstration experiments we will now analyze in detail a simple but physically relevant noise model for a single-qubit gate. Suppose that the gate is implemented by the noisy control Hamiltonian $H_c=J(t) \sigma_z$. Due to experimental imperfections the control $J(t)$ that is implemented is distinct from the nominal control $J_0(t)$ that would perfectly implement the required gate. Physically, this noise results in two distinct types of errors: \emph{dephasing}, where $\delta J(t)=J-J_0$ varies stochastically between uses of the gate, and \emph{calibration error} where $\delta J$ takes the same fixed value each time the gate is used. Where $\delta J(t)$ is stochastically varying we assume that the noise level does not change with time, and that that the noise spectrum for $\delta J(t)$ is mainly confined to frequencies $f>1/t_g$, where $t_g$ is the time required to perform the gate. When averaged over uses of the gate the resulting noisy operation is $\mathcal{E}\mathcal{U}$ where $\mathcal{U}$ is the desired gate and the noise process amounts to
\begin{equation}
\label{eq:CD}
	\mathcal{E}(\rho) = p \sigma_z e^{-i\delta \sigma_z}\rho e^{i\delta \sigma_z}\sigma_z + (1-p)e^{-i\delta \sigma_z}\rho e^{i\delta \sigma_z}.
\end{equation}
In this noise model the dephasing noise rate $p$ arises from the time-varying noise on the gate, while the unitary over rotation $\delta$ results from the fixed miscalibration of the control pulse $J(t)$. (Although we speak here in terms of calibration errors, this also approximately captures the effects of highly non-Markovian errors arising from very low-frequency noise in $J(t)$.) 

This noise model roughly captures many experimental gates, but more importantly it will demonstrate the range of behaviors that can be expected in terms of the relationship between average-case and worst-case error. Specifically when $\delta=0$ we have a pure dephasing process. For such errors~\cite{Magesan2012} the worst case error scales like $r$, so this is the most favorable possible behavior. On the other hand for $p=0$ we have a purely unitary rotation error that has the worst possible behavior, where the worst-case error scales like $\sqrt{r}$.

Using well-known techniques~\cite{Nielsen2002, Horodecki1999a} we find the average error rate for this calibration and dephasing (CD) noise to be $r_{\text{CD}} = \tfrac{2}{3}\left(p\cos(2\delta)+\sin^2\delta\right).$ Employing the semidefinite programming approach of Refs.~\cite{Watrous2009, Magesan2012}, we can evaluate the diamond distance for this noise channel and find $D_{\text{CD}} = \sqrt{\tfrac{3}{2}r_{\text{CD}}-p(1-p)}.$ A logarithmic plot of the ratio $D_{\text{CD}}/r_{\text{CD}}$ is shown in \autoref{fig:2dplot}.

\begin{figure}[t!]
\begin{center}
\includegraphics[width=\columnwidth]{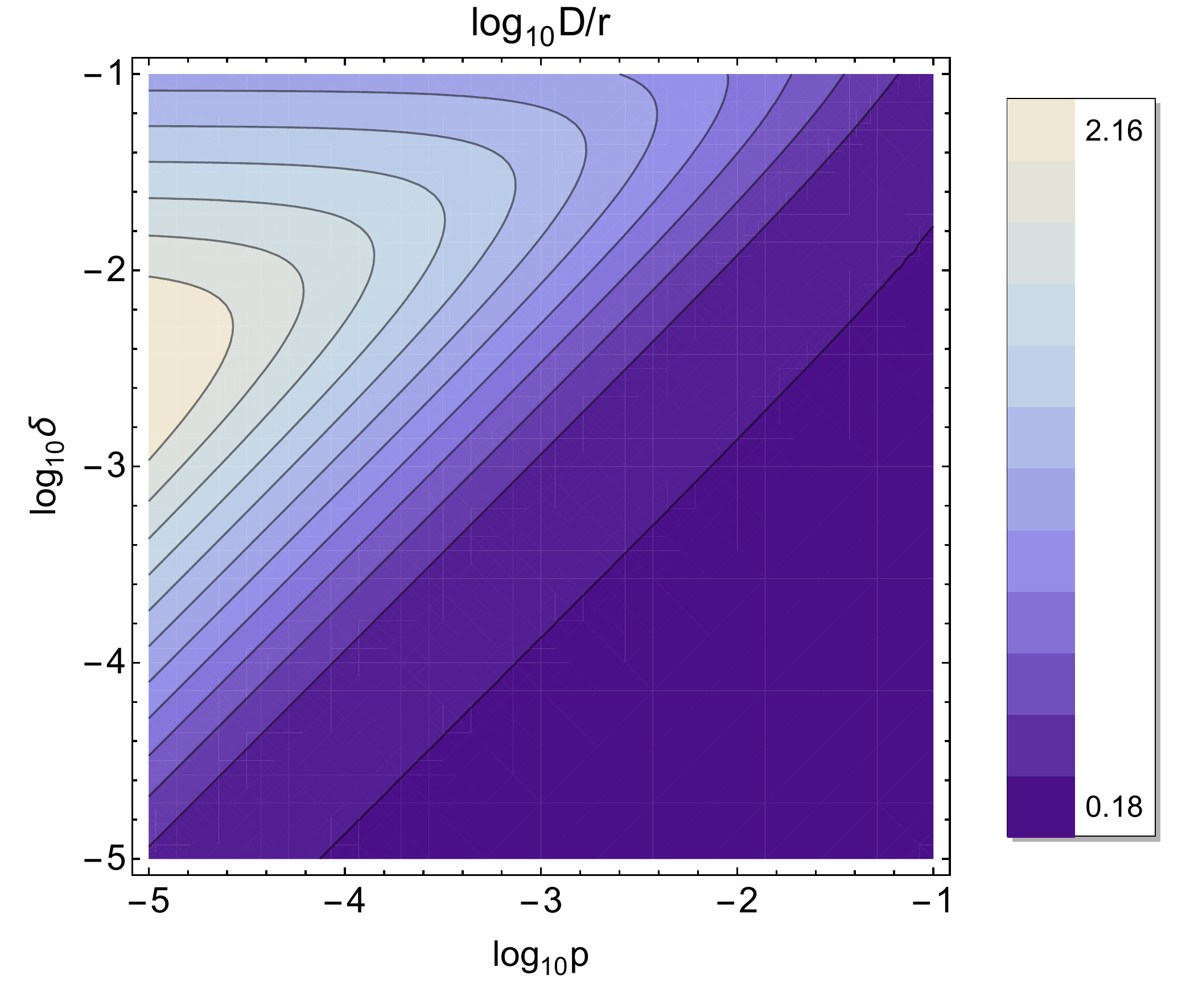}
\caption{Average error rate $r$ and worst-case error rate (diamond distance) $D$ for a combination of dephasing and unitary errors. The logarithmic plot is of $D/r$, which quantifies how much greater the worst-case error is than the average case as a function of a unitary over rotation angle $\delta$ and a dephasing probability $p$, where the exact noise process is given in Eq.~(\ref{eq:CD}). When $p\ge\delta$, then $D$ and $r$ are comparable to within a small factor, but as soon as $\delta > p$ then $D$ rapidly becomes much greater than $r$.}
\label{fig:2dplot}
\end{center}
\end{figure}

In the interesting regime of low error we find $r_{\text{CD}}\simeq 2(p+\delta^2)/3$, while $D_{\text{CD}}\simeq \sqrt{p^2+\delta^2}$. From this we can see that when $p\gg |\delta|$ we have $D_{\text{CD}}\simeq 3r_{\text{CD}}/2$, as for a pure dephasing process, and there is no great difference between worst-case and average-case errors. But as the calibration error grows, the worst-case error grows significantly. When calibration error dominates, $|\delta|\gg p$, we find $D_{\text{CD}}\simeq \sqrt{3r_{\text{CD}}/2}$. In this regime an average error rate $r_{\text{CD}}$ of around $10^{-4}$ corresponds to a more than one percent worst-case error. Physically then, we see that as dephasing error is reduced in a particular experimental setting, this places more stringent demands on the calibration required if the average error rate $r$ is to be compared directly to a fault-tolerance threshold.

\paragraph{Single-qubit relaxation errors.}

Another natural single-qubit noise process to consider is qubit relaxation or amplitude damping errors (spontaneous emission or a $T_1$ process in NMR language), at finite temperature. In this process a qubit with energy splitting $E$ is coupled to a bath at temperature $T$. 
Define as in~\cite{Nielsen2010} the probability for a decay process during the action of the gate is $\gamma p$ and the probability to go from the ground to the excited state is $\gamma (1-p)$. 
The ratio of upgoing to downgoing transition rates $p/(1-p) = \exp(-E/k_BT)$ is the Boltzmann factor, which allows us to identify $p=1/2$ as infinite temperature and $p=1$ as zero temperature. 
For this amplitude damping (AD) noise channel we find $r_{\text{AD}} = \bigl(1-\sqrt{1-\gamma}+\gamma/2\bigr)/3$. 
Although we have no closed form expression for the worst-case error for these channels, we have adapted standard techniques in the analysis of semidefinite programs to find the bound $D_{\text{AD}}\leq 3r_{\text{AD}} \max\{p,1-p\}$. Therefore we have a guarantee that the average-case and worst-case errors are not too different. Comparing with a direct evaluation of the semidefinite program we find $D_{\text{AD}}\simeq 3r_{\text{AD}}$ for zero temperature ($p=1$) and low noise $r_{\text{AD}}\ll 1$, so this is the tightest bound possible. In the limit of high temperature $p\rightarrow 1/2$ we approach a dephasing channel and recover the formula $D_{\text{AD}} = 3r_{\text{AD}}/2$. This behavior is illustrated in \autoref{fig:ampdamp}. 

\begin{figure}[t!]
\begin{center}
\includegraphics[width=\columnwidth]{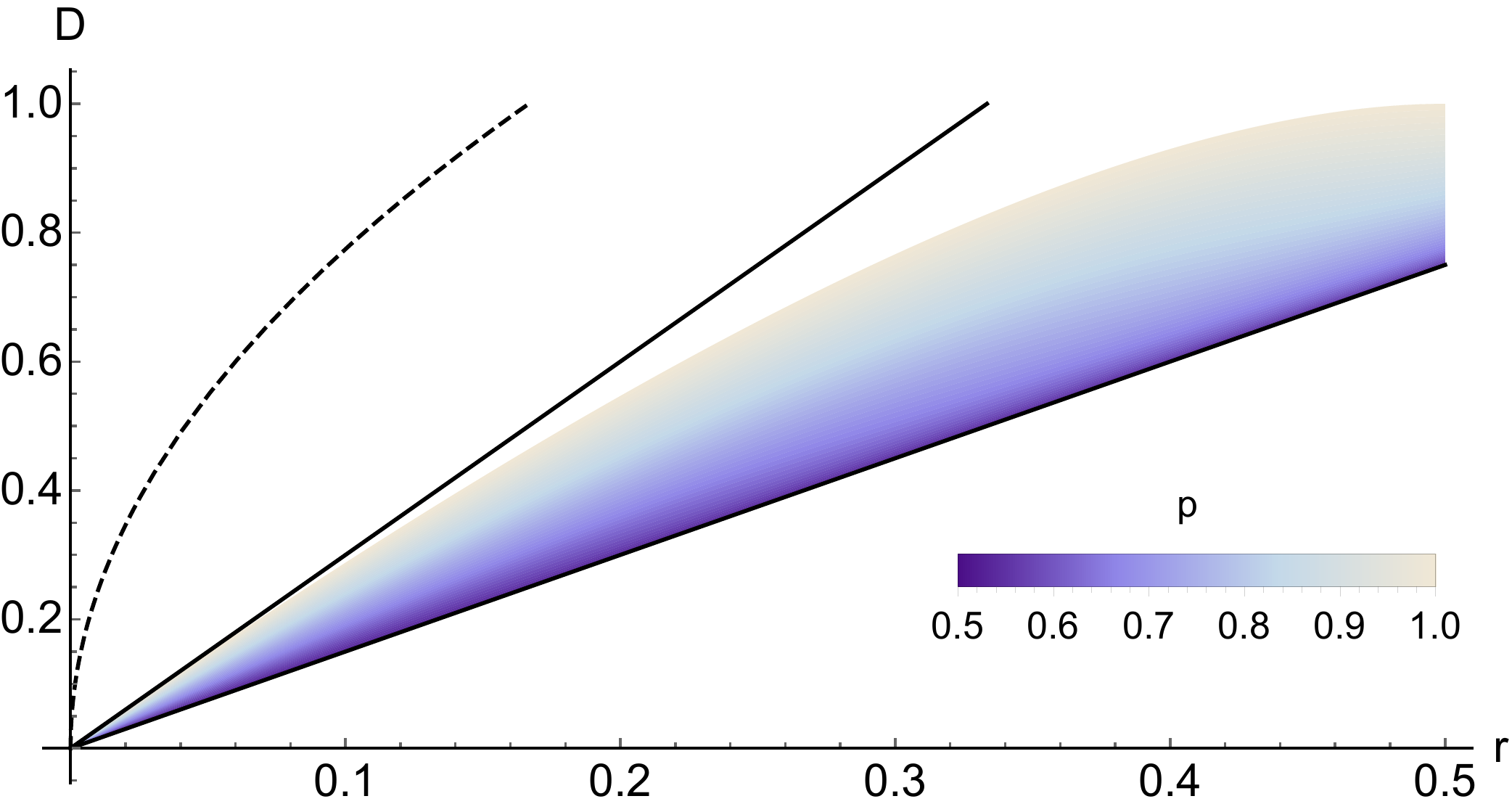}
\caption{Tradeoff between average error rate $r$ and the worst-case error rate in terms of the diamond distance $D$ for the thermal amplitude damping channel, where the parameter $p$ controls the temperature with $p=1$ corresponding to zero temperature and $p=1/2$ corresponding to infinite temperature. The dashed line is the previous best upper bound~\cite{Wallman2014}, while the upper black line is the new bound derived here. The zero-temperature limit ($p=1$) gives the least favorable scaling of $D$ with $r$, but in every case the bound $D \le 3r$ holds. The infinite-temperature limit ($p =1/2$) recovers the known value of $D=1.5 r$.}
\label{fig:ampdamp}
\end{center}
\end{figure}

\paragraph{Leakage errors}

Another important class of errors encountered in experiment is leakage errors. Modified randomized benchmarking protocols for leakage errors are proposed in~\cite{Epstein2014, Wallman2015a}. In Ref.~\cite{Epstein2014} it was shown that a nearly trivial modification of a standard benchmarking protocol in the presence of leakage errors can still be used to determine the average error rate $r$, so we again use this figure of merit for comparison. For a leakage model we need to consider a larger space of states, so we add a leakage level $|l\rangle$ to the two-qubit states $|0\rangle,|1\rangle$. We follow~\cite{Wallman2015a} in distinguishing coherent and incoherent leakage errors and compare the average-case error to the true worst-case error; this will be the diamond distance on the \emph{full} state space including both the leakage and qubit states. Fault-tolerance theorems also exist for leakage error processes~\cite{aliferisterhal2007} and this is the appropriate noise measure to compare with the numerical values found in those papers.

As an example of incoherent leakage (IL) we will consider the case where the qubit state $|1\rangle$ relaxes to $|l\rangle$ with probability $p$. A benchmarking experiment (following~\cite{Epstein2014}) then obtains the average-case error $r_{\text{IL}}=[1-\sqrt{1-p}+p]/3$ where this is now the infidelity averaged over states initially in the qubit subspace. Since this process is so similar to the amplitude damping channel we can use analogous techniques to find the inequality $D_{\text{IL}}\leq 2 r_{\text{IL}}$. Thus for this error process the average-case and worst-case error again almost coincide.

As an example of coherent leakage (CL), consider the unitary noise process $\E_{\text{CL}}(\rho)=U(\delta) \rho U (\delta)^\dagger$ given by $U (\delta) =\exp[-i\delta(|1\rangle\langle l|+|l\rangle\langle 1|)]$. For this noise process one obtains $r_{\text{CL}}=[1-\cos\delta -\cos^2\delta]/3$. However, as for the unitary errors discussed above, the worst-case error can be much larger than this. We find $\sqrt{3r_{\text{CL}}/2}\leq D_{\text{CL}} = |\sin\delta|\leq \sqrt{2 r_{\text{CL}}}$ for all $\delta \in [-\pi/2,\pi/2]$ and consequently the worst case error scales like $\sqrt{r_{\text{CL}}}$. Where leakage errors are possible, it would be important to use the methods of~\cite{Wallman2015a}, or some other method to bound coherent leakage errors, before comparing the average-case error $r$ to a fault-tolerance threshold.

\paragraph{Unitary errors.}

In looking at these examples we have found that unitary or nearly unitary errors appear to result in the largest difference between average-case and worst-case errors. This is true in general. For unitary errors in a $d$-dimensional space we find the following inequalities
\begin{equation*}
	\sqrt{\tfrac{d+1}{d}} \sqrt{r_{\text{U}}} \leq  D_{\text{U}}
	\leq \sqrt{(d+1)d}\sqrt{r_{\text{U}}}.
\end{equation*}
Thus any unitary error has a worst-case error scaling like $\sqrt{r_{\text{U}}}$.

\paragraph{A general inequality.}

For a large and important class of noise processes, the worst-case error can be \emph{directly} estimated from benchmarking-type data without side information about the type of error, which generally requires doing full quantum process tomography~\cite{Chuang1997}, or one of its SPAM-resistant variants~\cite{Kimmel2014, Blume-Kohout2013}. In Ref.~\cite{Wallman2015} a quantity called the unitarity $u(\E)$ of a noise process $\E$ was defined (see the Supplementary Material for a precise definition), and it was shown that this can be estimated efficiently and accurately using benchmarking. We find that for all unital noise (i.e.\ noise where the maximally mixed state is a fixed point) with no leakage, the unitarity and the average error rate together give a characterization of the worst-case error via the inequality~\footnote{A similar bound was recently derived independently by J.\ Wallman~\cite{Wallman2015b}.}
\begin{equation}
\label{eq:inequality}
	c_d \sqrt{u + \frac{2 d r}{d-1} -1 } 
	\leq D
	\leq d^2 c_d \sqrt{u + \frac{2dr}{d-1} -1 }\,,
\end{equation}
where $c_d = \frac{1}{2}(1-\frac{1}{d^2})^{1/2}$. Since the unitarity generally obeys the inequality $u \ge (1-dr/(d-1))^2$ (see Ref.~\cite{Wallman2015}) we find (for unital noise without leakage) that the worst-case error scaling matches the average-case \emph{if and only if} $u = 1-2dr/(d-1)+O(r^2)$. 

To illustrate the power of Inequality \eqref{eq:inequality}, we immediately find that for the single-qubit calibration and dephasing noise model, the condition $1-u_{\text{CD}} = 4r_{\text{CD}} + O(r_{\text{CD}}^2)$ is both necessary and sufficient to recover the favorable linear scaling between the worst- and average-case errors. In fact, the worst-case error for this channel can be expressed directly in terms of the unitarity as $D_{\text{CD}} = \sqrt{\frac{3}{2}r_{\text{CD}}-\frac{3}{8}(1-u_{\text{CD}})}$. And because the unitarity can be estimated from a benchmarking-type experiment, this gives direct experimental access to worst-case errors for this family of noise models without the need for expensive tomographic methods. 

Moreover, Inequality~\eqref{eq:inequality}  allows us to get insights into generalizing our conclusions for single-qubit models to few-qubit systems such as those required for entangling quantum gates. A natural generalization of our CD model to two-qubit calibration and dephasing errors would be an independent dephasing rate $p$ on each qubit and unitary noise given by $\mathrm{e}^{iH_{\text{CD2}}}$ where $H_{\text{CD2}} = \delta_1 \sigma_z^{(1)} + \delta_2 \sigma_z^{(2)} + \epsilon \sigma_z^{(1)}\sigma_z^{(2)}$.
The semidefinite programming approach is possible here, but becomes unwieldy because there are so many free parameters. However, both the average error rate and the unitarity are readily computed as in the appendix.
Inequality \eqref{eq:inequality} then allows one to easily and generally explore the tradeoffs in the calibration accuracy of the $\delta$ and $\epsilon$ parameters such that the overall error remains roughly consistent between average and worst case. Furthermore, since $u_{\text{CD2}}$ can be measured efficiently in a benchmarking experiment, large values of $u$ can be used to herald that an experiment has left the favorable scaling regime and more characterization and calibration must be done.

\paragraph{Conclusion and Outlook.}

We have seen that many realistic noise processes admit a linear relation between the average error rate (which is accessible experimentally) and the worst-case error (which is the relevant figure of merit for fault tolerance). 
The exceptions to this rule are highly coherent errors, where the worst-case error scales proportionally to the square root of the average error rate. 

While our methods and results are very general, there are noise sources that we have not tried to fit into our error taxonomy. Errors such as crosstalk~\cite{Gambetta2012}, time-dependent or non-Markovian noise~\cite{Fogarty2015, Ball2015} should be amenable to these methods, however, and extending our results to cover such noise is an important avenue for future work.

Finally, we reiterate that it is an interesting open question if it is possible to prove a fault-tolerance threshold result directly in terms of $r$ without the lossy conversion to $D$. Fault-tolerant circuits are not perfectly coherent since measuring error syndromes necessarily removes certain coherences, and this may provide an avenue to develop stronger theorems.

\paragraph{Acknowledgments.}

We thank R.\ Blume-Kohout, D.\ Gottesman, J. Gambetta, C.\ Granade and J.\ Wallman for discussions.
This work was supported by the ARC via EQuS project number CE11001013, 
by IARPA via the MQCO program, and by 
the US Army Research Office grant numbers W911NF-14-1-0098 and W911NF-14-1-0103. 
STF acknowledges support from an ARC Future Fellowship FT130101744.

\bibliography{worst_vs_avg}

\begin{thebibliography}{50}%
\makeatletter
\providecommand \@ifxundefined [1]{%
 \@ifx{#1\undefined}
}%
\providecommand \@ifnum [1]{%
 \ifnum #1\expandafter \@firstoftwo
 \else \expandafter \@secondoftwo
 \fi
}%
\providecommand \@ifx [1]{%
 \ifx #1\expandafter \@firstoftwo
 \else \expandafter \@secondoftwo
 \fi
}%
\providecommand \natexlab [1]{#1}%
\providecommand \enquote  [1]{``#1''}%
\providecommand \bibnamefont  [1]{#1}%
\providecommand \bibfnamefont [1]{#1}%
\providecommand \citenamefont [1]{#1}%
\providecommand \href@noop [0]{\@secondoftwo}%
\providecommand \href [0]{\begingroup \@sanitize@url \@href}%
\providecommand \@href[1]{\@@startlink{#1}\@@href}%
\providecommand \@@href[1]{\endgroup#1\@@endlink}%
\providecommand \@sanitize@url [0]{\catcode `\\12\catcode `\$12\catcode
  `\&12\catcode `\#12\catcode `\^12\catcode `\_12\catcode `\%12\relax}%
\providecommand \@@startlink[1]{}%
\providecommand \@@endlink[0]{}%
\providecommand \url  [0]{\begingroup\@sanitize@url \@url }%
\providecommand \@url [1]{\endgroup\@href {#1}{\urlprefix }}%
\providecommand \urlprefix  [0]{URL }%
\providecommand \Eprint [0]{\href }%
\providecommand \doibase [0]{http://dx.doi.org/}%
\providecommand \selectlanguage [0]{\@gobble}%
\providecommand \bibinfo  [0]{\@secondoftwo}%
\providecommand \bibfield  [0]{\@secondoftwo}%
\providecommand \translation [1]{[#1]}%
\providecommand \BibitemOpen [0]{}%
\providecommand \bibitemStop [0]{}%
\providecommand \bibitemNoStop [0]{.\EOS\space}%
\providecommand \EOS [0]{\spacefactor3000\relax}%
\providecommand \BibitemShut  [1]{\csname bibitem#1\endcsname}%
\let\auto@bib@innerbib\@empty
\bibitem [{\citenamefont {Kitaev}(1997)}]{Kitaev1997}%
  \BibitemOpen
  \bibfield  {author} {\bibinfo {author} {\bibfnamefont {A.~Y.}\ \bibnamefont
  {Kitaev}},\ }\href {\doibase 10.1070/RM1997v052n06ABEH002155} {\bibfield
  {journal} {\bibinfo  {journal} {Russian Math. Surv.}\ }\textbf {\bibinfo
  {volume} {52}},\ \bibinfo {pages} {1191} (\bibinfo {year}
  {1997})}\BibitemShut {NoStop}%
\bibitem [{\citenamefont {Aharonov}\ and\ \citenamefont
  {Ben-Or}(1997)}]{Aharonov1997}%
  \BibitemOpen
  \bibfield  {author} {\bibinfo {author} {\bibfnamefont {D.}~\bibnamefont
  {Aharonov}}\ and\ \bibinfo {author} {\bibfnamefont {M.}~\bibnamefont
  {Ben-Or}},\ }in\ \href {\doibase 10.1145/258533.258579} {\emph {\bibinfo
  {booktitle} {29th ACM Symp. on Theory of Computing (STOC)}}}\ (\bibinfo
  {address} {New York},\ \bibinfo {year} {1997})\ pp.\ \bibinfo {pages}
  {176--188},\ \Eprint {http://arxiv.org/abs/arXiv:quant-ph/9906129}
  {arXiv:quant-ph/9906129} \BibitemShut {NoStop}%
\bibitem [{\citenamefont {Knill}\ \emph {et~al.}(1998)\citenamefont {Knill},
  \citenamefont {Laflamme},\ and\ \citenamefont {Zurek}}]{Knill1998b}%
  \BibitemOpen
  \bibfield  {author} {\bibinfo {author} {\bibfnamefont {E.}~\bibnamefont
  {Knill}}, \bibinfo {author} {\bibfnamefont {R.}~\bibnamefont {Laflamme}}, \
  and\ \bibinfo {author} {\bibfnamefont {W.~H.}\ \bibnamefont {Zurek}},\ }\href
  {\doibase 10.1098/rspa.1998.0166} {\bibfield  {journal} {\bibinfo  {journal}
  {Proc. R. Soc. A}\ }\textbf {\bibinfo {volume} {454}},\ \bibinfo {pages}
  {365} (\bibinfo {year} {1998})},\ \Eprint
  {http://arxiv.org/abs/quant-ph/9702058} {arXiv:quant-ph/9702058} \BibitemShut
  {NoStop}%
\bibitem [{\citenamefont {Emerson}\ \emph {et~al.}(2005)\citenamefont
  {Emerson}, \citenamefont {Alicki},\ and\ \citenamefont
  {\.{Z}yczkowski}}]{Emerson2005}%
  \BibitemOpen
  \bibfield  {author} {\bibinfo {author} {\bibfnamefont {J.}~\bibnamefont
  {Emerson}}, \bibinfo {author} {\bibfnamefont {R.}~\bibnamefont {Alicki}}, \
  and\ \bibinfo {author} {\bibfnamefont {K.}~\bibnamefont {\.{Z}yczkowski}},\
  }\href {\doibase 10.1088/1464-4266/7/10/021} {\bibfield  {journal} {\bibinfo
  {journal} {J. Opt. B}\ }\textbf {\bibinfo {volume} {7}},\ \bibinfo {pages}
  {S347} (\bibinfo {year} {2005})},\ \Eprint
  {http://arxiv.org/abs/arXiv:quant-ph/0503243} {arXiv:quant-ph/0503243}
  \BibitemShut {NoStop}%
\bibitem [{\citenamefont {Knill}\ \emph {et~al.}(2008)\citenamefont {Knill},
  \citenamefont {Leibfried}, \citenamefont {Reichle}, \citenamefont {Britton},
  \citenamefont {Blakestad}, \citenamefont {Jost}, \citenamefont {Langer},
  \citenamefont {Ozeri}, \citenamefont {Seidelin},\ and\ \citenamefont
  {Wineland}}]{Knill2008}%
  \BibitemOpen
  \bibfield  {author} {\bibinfo {author} {\bibfnamefont {E.}~\bibnamefont
  {Knill}}, \bibinfo {author} {\bibfnamefont {D.}~\bibnamefont {Leibfried}},
  \bibinfo {author} {\bibfnamefont {R.}~\bibnamefont {Reichle}}, \bibinfo
  {author} {\bibfnamefont {J.}~\bibnamefont {Britton}}, \bibinfo {author}
  {\bibfnamefont {R.}~\bibnamefont {Blakestad}}, \bibinfo {author}
  {\bibfnamefont {J.}~\bibnamefont {Jost}}, \bibinfo {author} {\bibfnamefont
  {C.}~\bibnamefont {Langer}}, \bibinfo {author} {\bibfnamefont
  {R.}~\bibnamefont {Ozeri}}, \bibinfo {author} {\bibfnamefont
  {S.}~\bibnamefont {Seidelin}}, \ and\ \bibinfo {author} {\bibfnamefont
  {D.}~\bibnamefont {Wineland}},\ }\href {\doibase 10.1103/PhysRevA.77.012307}
  {\bibfield  {journal} {\bibinfo  {journal} {Phys. Rev. A}\ }\textbf {\bibinfo
  {volume} {77}},\ \bibinfo {pages} {012307} (\bibinfo {year} {2008})},\
  \Eprint {http://arxiv.org/abs/0707.0963} {arXiv:0707.0963} \BibitemShut
  {NoStop}%
\bibitem [{\citenamefont {Magesan}\ \emph
  {et~al.}(2012{\natexlab{a}})\citenamefont {Magesan}, \citenamefont
  {Gambetta}, \citenamefont {Johnson}, \citenamefont {Ryan}, \citenamefont
  {Chow}, \citenamefont {Merkel}, \citenamefont {da~Silva}, \citenamefont
  {Keefe}, \citenamefont {Rothwell}, \citenamefont {Ohki}, \citenamefont
  {Ketchen},\ and\ \citenamefont {Steffen}}]{Magesan2012b}%
  \BibitemOpen
  \bibfield  {author} {\bibinfo {author} {\bibfnamefont {E.}~\bibnamefont
  {Magesan}}, \bibinfo {author} {\bibfnamefont {J.~M.}\ \bibnamefont
  {Gambetta}}, \bibinfo {author} {\bibfnamefont {B.~R.}\ \bibnamefont
  {Johnson}}, \bibinfo {author} {\bibfnamefont {C.~A.}\ \bibnamefont {Ryan}},
  \bibinfo {author} {\bibfnamefont {J.~M.}\ \bibnamefont {Chow}}, \bibinfo
  {author} {\bibfnamefont {S.~T.}\ \bibnamefont {Merkel}}, \bibinfo {author}
  {\bibfnamefont {M.~P.}\ \bibnamefont {da~Silva}}, \bibinfo {author}
  {\bibfnamefont {G.~A.}\ \bibnamefont {Keefe}}, \bibinfo {author}
  {\bibfnamefont {M.~B.}\ \bibnamefont {Rothwell}}, \bibinfo {author}
  {\bibfnamefont {T.~A.}\ \bibnamefont {Ohki}}, \bibinfo {author}
  {\bibfnamefont {M.~B.}\ \bibnamefont {Ketchen}}, \ and\ \bibinfo {author}
  {\bibfnamefont {M.}~\bibnamefont {Steffen}},\ }\href {\doibase
  10.1103/PhysRevLett.109.080505} {\bibfield  {journal} {\bibinfo  {journal}
  {Phys. Rev. Lett.}\ }\textbf {\bibinfo {volume} {109}},\ \bibinfo {pages}
  {080505} (\bibinfo {year} {2012}{\natexlab{a}})},\ \Eprint
  {http://arxiv.org/abs/1203.4550} {arXiv:1203.4550} \BibitemShut {NoStop}%
\bibitem [{\citenamefont {Granade}\ \emph {et~al.}(2015)\citenamefont
  {Granade}, \citenamefont {Ferrie},\ and\ \citenamefont {Cory}}]{Granade2014}%
  \BibitemOpen
  \bibfield  {author} {\bibinfo {author} {\bibfnamefont {C.}~\bibnamefont
  {Granade}}, \bibinfo {author} {\bibfnamefont {C.}~\bibnamefont {Ferrie}}, \
  and\ \bibinfo {author} {\bibfnamefont {D.~G.}\ \bibnamefont {Cory}},\ }\href
  {\doibase 10.1088/1367-2630/17/1/013042} {\bibfield  {journal} {\bibinfo
  {journal} {New J. Phys.}\ }\textbf {\bibinfo {volume} {17}},\ \bibinfo
  {pages} {013042} (\bibinfo {year} {2015})},\ \Eprint
  {http://arxiv.org/abs/1404.5275} {arXiv:1404.5275} \BibitemShut {NoStop}%
\bibitem [{\citenamefont {Chow}\ \emph {et~al.}(2009)\citenamefont {Chow},
  \citenamefont {Gambetta}, \citenamefont {Tornberg}, \citenamefont {Koch},
  \citenamefont {Bishop}, \citenamefont {Houck}, \citenamefont {Johnson},
  \citenamefont {Frunzio}, \citenamefont {Girvin},\ and\ \citenamefont
  {Schoelkopf}}]{chow2009}%
  \BibitemOpen
  \bibfield  {author} {\bibinfo {author} {\bibfnamefont {J.~M.}\ \bibnamefont
  {Chow}}, \bibinfo {author} {\bibfnamefont {J.~M.}\ \bibnamefont {Gambetta}},
  \bibinfo {author} {\bibfnamefont {L.}~\bibnamefont {Tornberg}}, \bibinfo
  {author} {\bibfnamefont {J.}~\bibnamefont {Koch}}, \bibinfo {author}
  {\bibfnamefont {L.~S.}\ \bibnamefont {Bishop}}, \bibinfo {author}
  {\bibfnamefont {A.~A.}\ \bibnamefont {Houck}}, \bibinfo {author}
  {\bibfnamefont {B.~R.}\ \bibnamefont {Johnson}}, \bibinfo {author}
  {\bibfnamefont {L.}~\bibnamefont {Frunzio}}, \bibinfo {author} {\bibfnamefont
  {S.~M.}\ \bibnamefont {Girvin}}, \ and\ \bibinfo {author} {\bibfnamefont
  {R.~J.}\ \bibnamefont {Schoelkopf}},\ }\href {\doibase
  10.1103/PhysRevLett.102.090502} {\bibfield  {journal} {\bibinfo  {journal}
  {Phys. Rev. Lett.}\ }\textbf {\bibinfo {volume} {102}},\ \bibinfo {pages}
  {090502} (\bibinfo {year} {2009})},\ \Eprint {http://arxiv.org/abs/0811.4387}
  {arXiv:0811.4387} \BibitemShut {NoStop}%
\bibitem [{\citenamefont {Brown}\ \emph {et~al.}(2011)\citenamefont {Brown},
  \citenamefont {Wilson}, \citenamefont {Colombe}, \citenamefont {Ospelkaus},
  \citenamefont {Meier}, \citenamefont {Knill}, \citenamefont {Leibfried},\
  and\ \citenamefont {Wineland}}]{brown2011single}%
  \BibitemOpen
  \bibfield  {author} {\bibinfo {author} {\bibfnamefont {K.}~\bibnamefont
  {Brown}}, \bibinfo {author} {\bibfnamefont {A.}~\bibnamefont {Wilson}},
  \bibinfo {author} {\bibfnamefont {Y.}~\bibnamefont {Colombe}}, \bibinfo
  {author} {\bibfnamefont {C.}~\bibnamefont {Ospelkaus}}, \bibinfo {author}
  {\bibfnamefont {A.}~\bibnamefont {Meier}}, \bibinfo {author} {\bibfnamefont
  {E.}~\bibnamefont {Knill}}, \bibinfo {author} {\bibfnamefont
  {D.}~\bibnamefont {Leibfried}}, \ and\ \bibinfo {author} {\bibfnamefont
  {D.}~\bibnamefont {Wineland}},\ }\href {\doibase 10.1103/PhysRevA.84.030303}
  {\bibfield  {journal} {\bibinfo  {journal} {Phys. Rev. A}\ }\textbf {\bibinfo
  {volume} {84}},\ \bibinfo {pages} {030303} (\bibinfo {year} {2011})},\
  \Eprint {http://arxiv.org/abs/1104.2552} {arXiv:1104.2552} \BibitemShut
  {NoStop}%
\bibitem [{\citenamefont {Gaebler}\ \emph {et~al.}(2012)\citenamefont
  {Gaebler}, \citenamefont {Meier}, \citenamefont {Tan}, \citenamefont
  {Bowler}, \citenamefont {Lin}, \citenamefont {Hanneke}, \citenamefont {Jost},
  \citenamefont {Home}, \citenamefont {Knill}, \citenamefont {Leibfried} \emph
  {et~al.}}]{gaebler2012randomized}%
  \BibitemOpen
  \bibfield  {author} {\bibinfo {author} {\bibfnamefont {J.}~\bibnamefont
  {Gaebler}}, \bibinfo {author} {\bibfnamefont {A.}~\bibnamefont {Meier}},
  \bibinfo {author} {\bibfnamefont {T.~R.}\ \bibnamefont {Tan}}, \bibinfo
  {author} {\bibfnamefont {R.}~\bibnamefont {Bowler}}, \bibinfo {author}
  {\bibfnamefont {Y.}~\bibnamefont {Lin}}, \bibinfo {author} {\bibfnamefont
  {D.}~\bibnamefont {Hanneke}}, \bibinfo {author} {\bibfnamefont
  {J.}~\bibnamefont {Jost}}, \bibinfo {author} {\bibfnamefont {J.}~\bibnamefont
  {Home}}, \bibinfo {author} {\bibfnamefont {E.}~\bibnamefont {Knill}},
  \bibinfo {author} {\bibfnamefont {D.}~\bibnamefont {Leibfried}},  \emph
  {et~al.},\ }\href {\doibase 10.1103/PhysRevLett.108.260503} {\bibfield
  {journal} {\bibinfo  {journal} {Phy. Rev. Lett.}\ }\textbf {\bibinfo {volume}
  {108}},\ \bibinfo {pages} {260503} (\bibinfo {year} {2012})},\ \Eprint
  {http://arxiv.org/abs/1203.3733} {arXiv:1203.3733} \BibitemShut {NoStop}%
\bibitem [{\citenamefont {Chow}\ \emph {et~al.}(2012)\citenamefont {Chow},
  \citenamefont {Gambetta}, \citenamefont {C\'orcoles}, \citenamefont {Merkel},
  \citenamefont {Smolin}, \citenamefont {Rigetti}, \citenamefont {Poletto},
  \citenamefont {Keefe}, \citenamefont {Rothwell}, \citenamefont {Rozen},
  \citenamefont {Ketchen},\ and\ \citenamefont {Steffen}}]{chow2012}%
  \BibitemOpen
  \bibfield  {author} {\bibinfo {author} {\bibfnamefont {J.~M.}\ \bibnamefont
  {Chow}}, \bibinfo {author} {\bibfnamefont {J.~M.}\ \bibnamefont {Gambetta}},
  \bibinfo {author} {\bibfnamefont {A.~D.}\ \bibnamefont {C\'orcoles}},
  \bibinfo {author} {\bibfnamefont {S.~T.}\ \bibnamefont {Merkel}}, \bibinfo
  {author} {\bibfnamefont {J.~A.}\ \bibnamefont {Smolin}}, \bibinfo {author}
  {\bibfnamefont {C.}~\bibnamefont {Rigetti}}, \bibinfo {author} {\bibfnamefont
  {S.}~\bibnamefont {Poletto}}, \bibinfo {author} {\bibfnamefont {G.~A.}\
  \bibnamefont {Keefe}}, \bibinfo {author} {\bibfnamefont {M.~B.}\ \bibnamefont
  {Rothwell}}, \bibinfo {author} {\bibfnamefont {J.~R.}\ \bibnamefont {Rozen}},
  \bibinfo {author} {\bibfnamefont {M.~B.}\ \bibnamefont {Ketchen}}, \ and\
  \bibinfo {author} {\bibfnamefont {M.}~\bibnamefont {Steffen}},\ }\href
  {\doibase 10.1103/PhysRevLett.109.060501} {\bibfield  {journal} {\bibinfo
  {journal} {Phys. Rev. Lett.}\ }\textbf {\bibinfo {volume} {109}},\ \bibinfo
  {pages} {060501} (\bibinfo {year} {2012})},\ \Eprint
  {http://arxiv.org/abs/1202.5344} {arXiv:1202.5344} \BibitemShut {NoStop}%
\bibitem [{\citenamefont {C\'{o}rcoles}\ \emph {et~al.}(2013)\citenamefont
  {C\'{o}rcoles}, \citenamefont {Gambetta}, \citenamefont {Chow}, \citenamefont
  {Smolin}, \citenamefont {Ware}, \citenamefont {Strand}, \citenamefont
  {Plourde},\ and\ \citenamefont {Steffen}}]{corcoles2013}%
  \BibitemOpen
  \bibfield  {author} {\bibinfo {author} {\bibfnamefont {A.~D.}\ \bibnamefont
  {C\'{o}rcoles}}, \bibinfo {author} {\bibfnamefont {J.~M.}\ \bibnamefont
  {Gambetta}}, \bibinfo {author} {\bibfnamefont {J.~M.}\ \bibnamefont {Chow}},
  \bibinfo {author} {\bibfnamefont {J.~A.}\ \bibnamefont {Smolin}}, \bibinfo
  {author} {\bibfnamefont {M.}~\bibnamefont {Ware}}, \bibinfo {author}
  {\bibfnamefont {J.}~\bibnamefont {Strand}}, \bibinfo {author} {\bibfnamefont
  {B.~L.~T.}\ \bibnamefont {Plourde}}, \ and\ \bibinfo {author} {\bibfnamefont
  {M.}~\bibnamefont {Steffen}},\ }\href {\doibase 10.1103/PhysRevA.87.030301}
  {\bibfield  {journal} {\bibinfo  {journal} {Phys. Rev. A}\ }\textbf {\bibinfo
  {volume} {87}},\ \bibinfo {pages} {030301} (\bibinfo {year} {2013})},\
  \Eprint {http://arxiv.org/abs/1210.7011} {arXiv:1210.7011} \BibitemShut
  {NoStop}%
\bibitem [{\citenamefont {Chow}\ \emph {et~al.}(2014)\citenamefont {Chow},
  \citenamefont {Gambetta}, \citenamefont {Magesan}, \citenamefont {Abraham},
  \citenamefont {Cross}, \citenamefont {Johnson}, \citenamefont {Masluk},
  \citenamefont {Ryan}, \citenamefont {Smolin}, \citenamefont {Srinivasan}
  \emph {et~al.}}]{chow2014implementing}%
  \BibitemOpen
  \bibfield  {author} {\bibinfo {author} {\bibfnamefont {J.~M.}\ \bibnamefont
  {Chow}}, \bibinfo {author} {\bibfnamefont {J.~M.}\ \bibnamefont {Gambetta}},
  \bibinfo {author} {\bibfnamefont {E.}~\bibnamefont {Magesan}}, \bibinfo
  {author} {\bibfnamefont {D.~W.}\ \bibnamefont {Abraham}}, \bibinfo {author}
  {\bibfnamefont {A.~W.}\ \bibnamefont {Cross}}, \bibinfo {author}
  {\bibfnamefont {B.}~\bibnamefont {Johnson}}, \bibinfo {author} {\bibfnamefont
  {N.~A.}\ \bibnamefont {Masluk}}, \bibinfo {author} {\bibfnamefont {C.~A.}\
  \bibnamefont {Ryan}}, \bibinfo {author} {\bibfnamefont {J.~A.}\ \bibnamefont
  {Smolin}}, \bibinfo {author} {\bibfnamefont {S.~J.}\ \bibnamefont
  {Srinivasan}},  \emph {et~al.},\ }\href {\doibase 10.1038/ncomms5015}
  {\bibfield  {journal} {\bibinfo  {journal} {Nature Commun.}\ }\textbf
  {\bibinfo {volume} {5}} (\bibinfo {year} {2014}),\ 10.1038/ncomms5015},\
  \Eprint {http://arxiv.org/abs/1311.6330} {arXiv:1311.6330} \BibitemShut
  {NoStop}%
\bibitem [{\citenamefont {Barends}\ \emph {et~al.}(2014)\citenamefont
  {Barends}, \citenamefont {Kelly}, \citenamefont {Megrant}, \citenamefont
  {Veitia}, \citenamefont {Sank}, \citenamefont {Jeffrey}, \citenamefont
  {White}, \citenamefont {Mutus}, \citenamefont {Fowler}, \citenamefont
  {Campbell}, \citenamefont {Chen}, \citenamefont {Chen}, \citenamefont
  {Chiaro}, \citenamefont {Dunsworth}, \citenamefont {Neill}, \citenamefont
  {O'Malley}, \citenamefont {Roushan}, \citenamefont {Vainsencher},
  \citenamefont {Wenner}, \citenamefont {Korotkov}, \citenamefont {Cleland},\
  and\ \citenamefont {Martinis}}]{Barends2014}%
  \BibitemOpen
  \bibfield  {author} {\bibinfo {author} {\bibfnamefont {R.}~\bibnamefont
  {Barends}}, \bibinfo {author} {\bibfnamefont {J.}~\bibnamefont {Kelly}},
  \bibinfo {author} {\bibfnamefont {A.}~\bibnamefont {Megrant}}, \bibinfo
  {author} {\bibfnamefont {A.}~\bibnamefont {Veitia}}, \bibinfo {author}
  {\bibfnamefont {D.}~\bibnamefont {Sank}}, \bibinfo {author} {\bibfnamefont
  {E.}~\bibnamefont {Jeffrey}}, \bibinfo {author} {\bibfnamefont {T.~C.}\
  \bibnamefont {White}}, \bibinfo {author} {\bibfnamefont {J.}~\bibnamefont
  {Mutus}}, \bibinfo {author} {\bibfnamefont {A.~G.}\ \bibnamefont {Fowler}},
  \bibinfo {author} {\bibfnamefont {B.}~\bibnamefont {Campbell}}, \bibinfo
  {author} {\bibfnamefont {Y.}~\bibnamefont {Chen}}, \bibinfo {author}
  {\bibfnamefont {Z.}~\bibnamefont {Chen}}, \bibinfo {author} {\bibfnamefont
  {B.}~\bibnamefont {Chiaro}}, \bibinfo {author} {\bibfnamefont
  {A.}~\bibnamefont {Dunsworth}}, \bibinfo {author} {\bibfnamefont
  {C.}~\bibnamefont {Neill}}, \bibinfo {author} {\bibfnamefont
  {P.}~\bibnamefont {O'Malley}}, \bibinfo {author} {\bibfnamefont
  {P.}~\bibnamefont {Roushan}}, \bibinfo {author} {\bibfnamefont
  {A.}~\bibnamefont {Vainsencher}}, \bibinfo {author} {\bibfnamefont
  {J.}~\bibnamefont {Wenner}}, \bibinfo {author} {\bibfnamefont {A.~N.}\
  \bibnamefont {Korotkov}}, \bibinfo {author} {\bibfnamefont {A.~N.}\
  \bibnamefont {Cleland}}, \ and\ \bibinfo {author} {\bibfnamefont {J.~M.}\
  \bibnamefont {Martinis}},\ }\href {\doibase 10.1038/nature13171} {\bibfield
  {journal} {\bibinfo  {journal} {Nature}\ }\textbf {\bibinfo {volume} {508}},\
  \bibinfo {pages} {500} (\bibinfo {year} {2014})},\ \Eprint
  {http://arxiv.org/abs/1402.4848} {arXiv:1402.4848} \BibitemShut {NoStop}%
\bibitem [{\citenamefont {Harty}\ \emph {et~al.}(2014)\citenamefont {Harty},
  \citenamefont {Allcock}, \citenamefont {Ballance}, \citenamefont {Guidoni},
  \citenamefont {Janacek}, \citenamefont {Linke}, \citenamefont {Stacey},\ and\
  \citenamefont {Lucas}}]{harty2014high}%
  \BibitemOpen
  \bibfield  {author} {\bibinfo {author} {\bibfnamefont {T.}~\bibnamefont
  {Harty}}, \bibinfo {author} {\bibfnamefont {D.}~\bibnamefont {Allcock}},
  \bibinfo {author} {\bibfnamefont {C.}~\bibnamefont {Ballance}}, \bibinfo
  {author} {\bibfnamefont {L.}~\bibnamefont {Guidoni}}, \bibinfo {author}
  {\bibfnamefont {H.}~\bibnamefont {Janacek}}, \bibinfo {author} {\bibfnamefont
  {N.}~\bibnamefont {Linke}}, \bibinfo {author} {\bibfnamefont
  {D.}~\bibnamefont {Stacey}}, \ and\ \bibinfo {author} {\bibfnamefont
  {D.}~\bibnamefont {Lucas}},\ }\href {\doibase 10.1103/PhysRevLett.113.220501}
  {\bibfield  {journal} {\bibinfo  {journal} {Phys. Rev. Lett.}\ }\textbf
  {\bibinfo {volume} {113}},\ \bibinfo {pages} {220501} (\bibinfo {year}
  {2014})},\ \Eprint {http://arxiv.org/abs/1403.1524} {arXiv:1403.1524}
  \BibitemShut {NoStop}%
\bibitem [{\citenamefont {Muhonen}\ \emph {et~al.}(2015)\citenamefont
  {Muhonen}, \citenamefont {Laucht}, \citenamefont {Simmons}, \citenamefont
  {Dehollain}, \citenamefont {Kalra}, \citenamefont {Hudson}, \citenamefont
  {Freer}, \citenamefont {Itoh}, \citenamefont {Jamieson}, \citenamefont
  {McCallum}, \citenamefont {Dzurak},\ and\ \citenamefont
  {Morello}}]{muhonen2015}%
  \BibitemOpen
  \bibfield  {author} {\bibinfo {author} {\bibfnamefont {J.~T.}\ \bibnamefont
  {Muhonen}}, \bibinfo {author} {\bibfnamefont {A.}~\bibnamefont {Laucht}},
  \bibinfo {author} {\bibfnamefont {S.}~\bibnamefont {Simmons}}, \bibinfo
  {author} {\bibfnamefont {J.~P.}\ \bibnamefont {Dehollain}}, \bibinfo {author}
  {\bibfnamefont {R.}~\bibnamefont {Kalra}}, \bibinfo {author} {\bibfnamefont
  {F.~E.}\ \bibnamefont {Hudson}}, \bibinfo {author} {\bibfnamefont
  {S.}~\bibnamefont {Freer}}, \bibinfo {author} {\bibfnamefont {K.~M.}\
  \bibnamefont {Itoh}}, \bibinfo {author} {\bibfnamefont {D.~N.}\ \bibnamefont
  {Jamieson}}, \bibinfo {author} {\bibfnamefont {J.~C.}\ \bibnamefont
  {McCallum}}, \bibinfo {author} {\bibfnamefont {A.~S.}\ \bibnamefont
  {Dzurak}}, \ and\ \bibinfo {author} {\bibfnamefont {A.}~\bibnamefont
  {Morello}},\ }\href {\doibase 10.1088/0953-8984/27/15/154205} {\bibfield
  {journal} {\bibinfo  {journal} {J. Phys.: Condens. Matter}\ }\textbf
  {\bibinfo {volume} {27}},\ \bibinfo {pages} {154205} (\bibinfo {year}
  {2015})},\ \Eprint {http://arxiv.org/abs/1410.2338} {arXiv:1410.2338}
  \BibitemShut {NoStop}%
\bibitem [{\citenamefont {Xia}\ \emph {et~al.}(2015)\citenamefont {Xia},
  \citenamefont {Lichtman}, \citenamefont {Maller}, \citenamefont {Carr},
  \citenamefont {Piotrowicz}, \citenamefont {Isenhower},\ and\ \citenamefont
  {Saffman}}]{xia2015}%
  \BibitemOpen
  \bibfield  {author} {\bibinfo {author} {\bibfnamefont {T.}~\bibnamefont
  {Xia}}, \bibinfo {author} {\bibfnamefont {M.}~\bibnamefont {Lichtman}},
  \bibinfo {author} {\bibfnamefont {K.}~\bibnamefont {Maller}}, \bibinfo
  {author} {\bibfnamefont {A.~W.}\ \bibnamefont {Carr}}, \bibinfo {author}
  {\bibfnamefont {M.~J.}\ \bibnamefont {Piotrowicz}}, \bibinfo {author}
  {\bibfnamefont {L.}~\bibnamefont {Isenhower}}, \ and\ \bibinfo {author}
  {\bibfnamefont {M.}~\bibnamefont {Saffman}},\ }\href {\doibase
  10.1103/PhysRevLett.114.100503} {\bibfield  {journal} {\bibinfo  {journal}
  {Phys. Rev. Lett.}\ }\textbf {\bibinfo {volume} {114}},\ \bibinfo {pages}
  {100503} (\bibinfo {year} {2015})},\ \Eprint
  {http://arxiv.org/abs/1501.02041} {arXiv:1501.02041} \BibitemShut {NoStop}%
\bibitem [{\citenamefont {Sanders}\ \emph {et~al.}(2015)\citenamefont
  {Sanders}, \citenamefont {Wallman},\ and\ \citenamefont
  {Sanders}}]{Sanders2015}%
  \BibitemOpen
  \bibfield  {author} {\bibinfo {author} {\bibfnamefont {Y.~R.}\ \bibnamefont
  {Sanders}}, \bibinfo {author} {\bibfnamefont {J.~J.}\ \bibnamefont
  {Wallman}}, \ and\ \bibinfo {author} {\bibfnamefont {B.~C.}\ \bibnamefont
  {Sanders}},\ }\href {\doibase 10.1088/1367-2630/18/1/012002} {\bibfield
  {journal} {\bibinfo  {journal} {New J. Phys.}\ }\textbf {\bibinfo {volume}
  {18}},\ \bibinfo {pages} {012002} (\bibinfo {year} {2015})},\ \Eprint
  {http://arxiv.org/abs/1501.04932} {arXiv:1501.04932} \BibitemShut {NoStop}%
\bibitem [{\citenamefont {Magesan}\ \emph
  {et~al.}(2012{\natexlab{b}})\citenamefont {Magesan}, \citenamefont
  {Gambetta},\ and\ \citenamefont {Emerson}}]{Magesan2012}%
  \BibitemOpen
  \bibfield  {author} {\bibinfo {author} {\bibfnamefont {E.}~\bibnamefont
  {Magesan}}, \bibinfo {author} {\bibfnamefont {J.~M.}\ \bibnamefont
  {Gambetta}}, \ and\ \bibinfo {author} {\bibfnamefont {J.}~\bibnamefont
  {Emerson}},\ }\href {\doibase 10.1103/PhysRevA.85.042311} {\bibfield
  {journal} {\bibinfo  {journal} {Phys. Rev. A}\ }\textbf {\bibinfo {volume}
  {85}},\ \bibinfo {pages} {042311} (\bibinfo {year} {2012}{\natexlab{b}})},\
  \Eprint {http://arxiv.org/abs/1109.6887} {arXiv:1109.6887} \BibitemShut
  {NoStop}%
\bibitem [{\citenamefont {Aliferis}\ \emph {et~al.}(2006)\citenamefont
  {Aliferis}, \citenamefont {Gottesman},\ and\ \citenamefont
  {Preskill}}]{Aliferis2006a}%
  \BibitemOpen
  \bibfield  {author} {\bibinfo {author} {\bibfnamefont {P.}~\bibnamefont
  {Aliferis}}, \bibinfo {author} {\bibfnamefont {D.}~\bibnamefont {Gottesman}},
  \ and\ \bibinfo {author} {\bibfnamefont {J.}~\bibnamefont {Preskill}},\
  }\href@noop {} {\bibfield  {journal} {\bibinfo  {journal} {Quant. Inf.
  Comput.}\ }\textbf {\bibinfo {volume} {6}},\ \bibinfo {pages} {97} (\bibinfo
  {year} {2006})},\ \Eprint {http://arxiv.org/abs/arXiv:quant-ph/0504218}
  {arXiv:quant-ph/0504218} \BibitemShut {NoStop}%
\bibitem [{\citenamefont {Aliferis}\ and\ \citenamefont
  {Cross}(2007)}]{Aliferis2007}%
  \BibitemOpen
  \bibfield  {author} {\bibinfo {author} {\bibfnamefont {P.}~\bibnamefont
  {Aliferis}}\ and\ \bibinfo {author} {\bibfnamefont {A.~W.}\ \bibnamefont
  {Cross}},\ }\href {\doibase 10.1103/PhysRevLett.98.220502} {\bibfield
  {journal} {\bibinfo  {journal} {Phys. Rev. Lett.}\ }\textbf {\bibinfo
  {volume} {98}},\ \bibinfo {pages} {220502} (\bibinfo {year} {2007})},\
  \Eprint {http://arxiv.org/abs/arXiv:quant-ph/0610063}
  {arXiv:quant-ph/0610063} \BibitemShut {NoStop}%
\bibitem [{\citenamefont {Aliferis}\ \emph {et~al.}(2009)\citenamefont
  {Aliferis}, \citenamefont {Brito}, \citenamefont {DiVincenzo}, \citenamefont
  {Preskill}, \citenamefont {Steffen},\ and\ \citenamefont
  {Terhal}}]{Aliferis2009}%
  \BibitemOpen
  \bibfield  {author} {\bibinfo {author} {\bibfnamefont {P.}~\bibnamefont
  {Aliferis}}, \bibinfo {author} {\bibfnamefont {F.}~\bibnamefont {Brito}},
  \bibinfo {author} {\bibfnamefont {D.~P.}\ \bibnamefont {DiVincenzo}},
  \bibinfo {author} {\bibfnamefont {J.}~\bibnamefont {Preskill}}, \bibinfo
  {author} {\bibfnamefont {M.}~\bibnamefont {Steffen}}, \ and\ \bibinfo
  {author} {\bibfnamefont {B.~M.}\ \bibnamefont {Terhal}},\ }\href {\doibase
  10.1088/1367-2630/11/1/013061} {\bibfield  {journal} {\bibinfo  {journal}
  {New J. Phys.}\ }\textbf {\bibinfo {volume} {11}},\ \bibinfo {pages} {013061}
  (\bibinfo {year} {2009})},\ \Eprint {http://arxiv.org/abs/0806.0383}
  {arXiv:0806.0383} \BibitemShut {NoStop}%
\bibitem [{\citenamefont {Wang}\ \emph {et~al.}(2011)\citenamefont {Wang},
  \citenamefont {Fowler},\ and\ \citenamefont {Hollenberg}}]{Wang2011}%
  \BibitemOpen
  \bibfield  {author} {\bibinfo {author} {\bibfnamefont {D.~S.}\ \bibnamefont
  {Wang}}, \bibinfo {author} {\bibfnamefont {A.~G.}\ \bibnamefont {Fowler}}, \
  and\ \bibinfo {author} {\bibfnamefont {L.~C.~L.}\ \bibnamefont
  {Hollenberg}},\ }\href {\doibase 10.1103/PhysRevA.83.020302} {\bibfield
  {journal} {\bibinfo  {journal} {Phys. Rev. A}\ }\textbf {\bibinfo {volume}
  {83}},\ \bibinfo {pages} {020302} (\bibinfo {year} {2011})},\ \Eprint
  {http://arxiv.org/abs/1009.3686} {arXiv:1009.3686} \BibitemShut {NoStop}%
\bibitem [{\citenamefont {Gilchrist}\ \emph {et~al.}(2005)\citenamefont
  {Gilchrist}, \citenamefont {Langford},\ and\ \citenamefont
  {Nielsen}}]{Gilchrist2005}%
  \BibitemOpen
  \bibfield  {author} {\bibinfo {author} {\bibfnamefont {A.}~\bibnamefont
  {Gilchrist}}, \bibinfo {author} {\bibfnamefont {N.~K.}\ \bibnamefont
  {Langford}}, \ and\ \bibinfo {author} {\bibfnamefont {M.~A.}\ \bibnamefont
  {Nielsen}},\ }\href {\doibase 10.1103/PhysRevA.71.062310} {\bibfield
  {journal} {\bibinfo  {journal} {Phys. Rev. A}\ }\textbf {\bibinfo {volume}
  {71}},\ \bibinfo {pages} {062310} (\bibinfo {year} {2005})},\ \Eprint
  {http://arxiv.org/abs/arXiv:quant-ph/0408063} {arXiv:quant-ph/0408063}
  \BibitemShut {NoStop}%
\bibitem [{\citenamefont {Helstrom}(1967)}]{Helstrom1967}%
  \BibitemOpen
  \bibfield  {author} {\bibinfo {author} {\bibfnamefont {C.~W.}\ \bibnamefont
  {Helstrom}},\ }\href {\doibase
  http://dx.doi.org/10.1016/S0019-9958(67)90302-6} {\bibfield  {journal}
  {\bibinfo  {journal} {Information and Control}\ }\textbf {\bibinfo {volume}
  {10}},\ \bibinfo {pages} {254 } (\bibinfo {year} {1967})}\BibitemShut
  {NoStop}%
\bibitem [{\citenamefont {Watrous}(2009)}]{Watrous2009}%
  \BibitemOpen
  \bibfield  {author} {\bibinfo {author} {\bibfnamefont {J.}~\bibnamefont
  {Watrous}},\ }\href {\doibase 10.4086/toc.2009.v005a011} {\bibfield
  {journal} {\bibinfo  {journal} {Theory of Computation}\ }\textbf {\bibinfo
  {volume} {5}},\ \bibinfo {pages} {217} (\bibinfo {year} {2009})},\ \Eprint
  {http://arxiv.org/abs/0901.4709} {arXiv:0901.4709} \BibitemShut {NoStop}%
\bibitem [{\citenamefont {Ben-Aroya}\ and\ \citenamefont
  {Ta-Shma}(2010)}]{Ben-Aroya2010}%
  \BibitemOpen
  \bibfield  {author} {\bibinfo {author} {\bibfnamefont {A.}~\bibnamefont
  {Ben-Aroya}}\ and\ \bibinfo {author} {\bibfnamefont {A.}~\bibnamefont
  {Ta-Shma}},\ }\href@noop {} {\bibfield  {journal} {\bibinfo  {journal}
  {Quant. Inf. Comp.}\ }\textbf {\bibinfo {volume} {10}},\ \bibinfo {pages}
  {77} (\bibinfo {year} {2010})},\ \Eprint {http://arxiv.org/abs/0902.3397}
  {arXiv:0902.3397} \BibitemShut {NoStop}%
\bibitem [{\citenamefont {Watrous}(2013)}]{Watrous2013}%
  \BibitemOpen
  \bibfield  {author} {\bibinfo {author} {\bibfnamefont {J.}~\bibnamefont
  {Watrous}},\ }\href {\doibase 10.4086/cjtcs.2013.008} {\bibfield  {journal}
  {\bibinfo  {journal} {Chicago J. Theo. Comp. Sci.}\ }\textbf {\bibinfo
  {volume} {2013}},\ \bibinfo {pages} {1} (\bibinfo {year} {2013})},\ \Eprint
  {http://arxiv.org/abs/1207.5726} {arXiv:1207.5726} \BibitemShut {NoStop}%
\bibitem [{\citenamefont {Wallman}\ and\ \citenamefont
  {Flammia}(2014)}]{Wallman2014}%
  \BibitemOpen
  \bibfield  {author} {\bibinfo {author} {\bibfnamefont {J.~J.}\ \bibnamefont
  {Wallman}}\ and\ \bibinfo {author} {\bibfnamefont {S.~T.}\ \bibnamefont
  {Flammia}},\ }\href {\doibase 10.1088/1367-2630/16/10/103032} {\bibfield
  {journal} {\bibinfo  {journal} {New J. Phys.}\ }\textbf {\bibinfo {volume}
  {16}},\ \bibinfo {pages} {103032} (\bibinfo {year} {2014})},\ \Eprint
  {http://arxiv.org/abs/1404.6025} {arXiv:1404.6025} \BibitemShut {NoStop}%
\bibitem [{\citenamefont {Nielsen}(2002)}]{Nielsen2002}%
  \BibitemOpen
  \bibfield  {author} {\bibinfo {author} {\bibfnamefont {M.~A.}\ \bibnamefont
  {Nielsen}},\ }\href {\doibase 10.1016/S0375-9601(02)01272-0} {\bibfield
  {journal} {\bibinfo  {journal} {Phys. Lett. A}\ }\textbf {\bibinfo {volume}
  {303}},\ \bibinfo {pages} {249} (\bibinfo {year} {2002})},\ \Eprint
  {http://arxiv.org/abs/arXiv:quant-ph/0205035} {arXiv:quant-ph/0205035}
  \BibitemShut {NoStop}%
\bibitem [{\citenamefont {Horodecki}\ \emph {et~al.}(1999)\citenamefont
  {Horodecki}, \citenamefont {Horodecki},\ and\ \citenamefont
  {Horodecki}}]{Horodecki1999a}%
  \BibitemOpen
  \bibfield  {author} {\bibinfo {author} {\bibfnamefont {M.}~\bibnamefont
  {Horodecki}}, \bibinfo {author} {\bibfnamefont {P.}~\bibnamefont
  {Horodecki}}, \ and\ \bibinfo {author} {\bibfnamefont {R.}~\bibnamefont
  {Horodecki}},\ }\href {\doibase 10.1103/PhysRevA.60.1888} {\bibfield
  {journal} {\bibinfo  {journal} {Phys. Rev. A}\ }\textbf {\bibinfo {volume}
  {60}},\ \bibinfo {pages} {1888} (\bibinfo {year} {1999})},\ \Eprint
  {http://arxiv.org/abs/quant-ph/9807091} {quant-ph/9807091} \BibitemShut
  {NoStop}%
\bibitem [{\citenamefont {Nielsen}\ and\ \citenamefont
  {Chuang}(2000)}]{Nielsen2010}%
  \BibitemOpen
  \bibfield  {author} {\bibinfo {author} {\bibfnamefont {M.~A.}\ \bibnamefont
  {Nielsen}}\ and\ \bibinfo {author} {\bibfnamefont {I.~L.}\ \bibnamefont
  {Chuang}},\ }\href@noop {} {\emph {\bibinfo {title} {{Quantum Computation and
  Quantum Information}}}}\ (\bibinfo  {publisher} {Cambridge University
  Press},\ \bibinfo {address} {New York},\ \bibinfo {year} {2000})\BibitemShut
  {NoStop}%
\bibitem [{\citenamefont {Epstein}\ \emph {et~al.}(2014)\citenamefont
  {Epstein}, \citenamefont {Cross}, \citenamefont {Magesan},\ and\
  \citenamefont {Gambetta}}]{Epstein2014}%
  \BibitemOpen
  \bibfield  {author} {\bibinfo {author} {\bibfnamefont {J.~M.}\ \bibnamefont
  {Epstein}}, \bibinfo {author} {\bibfnamefont {A.~W.}\ \bibnamefont {Cross}},
  \bibinfo {author} {\bibfnamefont {E.}~\bibnamefont {Magesan}}, \ and\
  \bibinfo {author} {\bibfnamefont {J.~M.}\ \bibnamefont {Gambetta}},\ }\href
  {\doibase 10.1103/PhysRevA.89.062321} {\bibfield  {journal} {\bibinfo
  {journal} {Phys. Rev. A}\ }\textbf {\bibinfo {volume} {89}},\ \bibinfo
  {pages} {062321} (\bibinfo {year} {2014})},\ \Eprint
  {http://arxiv.org/abs/1308.2928} {arXiv:1308.2928} \BibitemShut {NoStop}%
\bibitem [{\citenamefont {Wallman}\ \emph
  {et~al.}(2015{\natexlab{a}})\citenamefont {Wallman}, \citenamefont
  {Barnhill},\ and\ \citenamefont {Emerson}}]{Wallman2015a}%
  \BibitemOpen
  \bibfield  {author} {\bibinfo {author} {\bibfnamefont {J.~J.}\ \bibnamefont
  {Wallman}}, \bibinfo {author} {\bibfnamefont {M.}~\bibnamefont {Barnhill}}, \
  and\ \bibinfo {author} {\bibfnamefont {J.}~\bibnamefont {Emerson}},\ }\href
  {\doibase 10.1103/PhysRevLett.115.060501} {\bibfield  {journal} {\bibinfo
  {journal} {Phys. Rev. Lett.}\ }\textbf {\bibinfo {volume} {115}},\ \bibinfo
  {pages} {060501} (\bibinfo {year} {2015}{\natexlab{a}})},\ \Eprint
  {http://arxiv.org/abs/1412.4126} {arXiv:1412.4126} \BibitemShut {NoStop}%
\bibitem [{\citenamefont {Aliferis}\ and\ \citenamefont
  {Terhal}(2007)}]{aliferisterhal2007}%
  \BibitemOpen
  \bibfield  {author} {\bibinfo {author} {\bibfnamefont {P.}~\bibnamefont
  {Aliferis}}\ and\ \bibinfo {author} {\bibfnamefont {B.~M.}\ \bibnamefont
  {Terhal}},\ }\href@noop {} {\bibfield  {journal} {\bibinfo  {journal} {Quant.
  Inf. Comp.}\ }\textbf {\bibinfo {volume} {7}},\ \bibinfo {pages} {139}
  (\bibinfo {year} {2007})},\ \Eprint {http://arxiv.org/abs/quant-ph/0511065}
  {quant-ph/0511065} \BibitemShut {NoStop}%
\bibitem [{\citenamefont {Chuang}\ and\ \citenamefont
  {Nielsen}(1997)}]{Chuang1997}%
  \BibitemOpen
  \bibfield  {author} {\bibinfo {author} {\bibfnamefont {I.~L.}\ \bibnamefont
  {Chuang}}\ and\ \bibinfo {author} {\bibfnamefont {M.~A.}\ \bibnamefont
  {Nielsen}},\ }\href {\doibase 10.1080/09500349708231894} {\bibfield
  {journal} {\bibinfo  {journal} {J. Mod. Opt.}\ }\textbf {\bibinfo {volume}
  {44}},\ \bibinfo {pages} {2455} (\bibinfo {year} {1997})},\ \Eprint
  {http://arxiv.org/abs/arXiv:quant-ph/9610001} {arXiv:quant-ph/9610001}
  \BibitemShut {NoStop}%
\bibitem [{\citenamefont {Kimmel}\ \emph {et~al.}(2014)\citenamefont {Kimmel},
  \citenamefont {da~Silva}, \citenamefont {Ryan}, \citenamefont {Johnson},\
  and\ \citenamefont {Ohki}}]{Kimmel2014}%
  \BibitemOpen
  \bibfield  {author} {\bibinfo {author} {\bibfnamefont {S.}~\bibnamefont
  {Kimmel}}, \bibinfo {author} {\bibfnamefont {M.~P.}\ \bibnamefont
  {da~Silva}}, \bibinfo {author} {\bibfnamefont {C.~A.}\ \bibnamefont {Ryan}},
  \bibinfo {author} {\bibfnamefont {B.~R.}\ \bibnamefont {Johnson}}, \ and\
  \bibinfo {author} {\bibfnamefont {T.}~\bibnamefont {Ohki}},\ }\href {\doibase
  10.1103/PhysRevX.4.011050} {\bibfield  {journal} {\bibinfo  {journal} {Phys.
  Rev. X}\ }\textbf {\bibinfo {volume} {4}},\ \bibinfo {pages} {011050}
  (\bibinfo {year} {2014})},\ \Eprint {http://arxiv.org/abs/1306.2348}
  {arXiv:1306.2348} \BibitemShut {NoStop}%
\bibitem [{\citenamefont {Blume-Kohout}\ \emph {et~al.}(2013)\citenamefont
  {Blume-Kohout}, \citenamefont {Gamble}, \citenamefont {Nielsen},
  \citenamefont {Mizrahi}, \citenamefont {Sterk},\ and\ \citenamefont
  {Maunz}}]{Blume-Kohout2013}%
  \BibitemOpen
  \bibfield  {author} {\bibinfo {author} {\bibfnamefont {R.}~\bibnamefont
  {Blume-Kohout}}, \bibinfo {author} {\bibfnamefont {J.~K.}\ \bibnamefont
  {Gamble}}, \bibinfo {author} {\bibfnamefont {E.}~\bibnamefont {Nielsen}},
  \bibinfo {author} {\bibfnamefont {J.}~\bibnamefont {Mizrahi}}, \bibinfo
  {author} {\bibfnamefont {J.~D.}\ \bibnamefont {Sterk}}, \ and\ \bibinfo
  {author} {\bibfnamefont {P.}~\bibnamefont {Maunz}},\ }\href@noop {} {\
  (\bibinfo {year} {2013})},\ \Eprint {http://arxiv.org/abs/1310.4492}
  {arXiv:1310.4492} \BibitemShut {NoStop}%
\bibitem [{\citenamefont {Wallman}\ \emph
  {et~al.}(2015{\natexlab{b}})\citenamefont {Wallman}, \citenamefont {Granade},
  \citenamefont {Harper},\ and\ \citenamefont {Flammia}}]{Wallman2015}%
  \BibitemOpen
  \bibfield  {author} {\bibinfo {author} {\bibfnamefont {J.}~\bibnamefont
  {Wallman}}, \bibinfo {author} {\bibfnamefont {C.}~\bibnamefont {Granade}},
  \bibinfo {author} {\bibfnamefont {R.}~\bibnamefont {Harper}}, \ and\ \bibinfo
  {author} {\bibfnamefont {S.~T.}\ \bibnamefont {Flammia}},\ }\href {\doibase
  10.1088/1367-2630/17/11/113020} {\bibfield  {journal} {\bibinfo  {journal}
  {New J. Phys.}\ }\textbf {\bibinfo {volume} {17}},\ \bibinfo {pages} {113020}
  (\bibinfo {year} {2015}{\natexlab{b}})},\ \Eprint
  {http://arxiv.org/abs/1503.07865} {arXiv:1503.07865} \BibitemShut {NoStop}%
\bibitem [{Note1()}]{Note1}%
  \BibitemOpen
  \bibinfo {note} {A similar bound was recently derived independently by J.\
  Wallman~\cite {Wallman2015b}.}\BibitemShut {Stop}%
\bibitem [{\citenamefont {Gambetta}\ \emph {et~al.}(2012)\citenamefont
  {Gambetta}, \citenamefont {C\'{o}rcoles}, \citenamefont {Merkel},
  \citenamefont {Johnson}, \citenamefont {Smolin}, \citenamefont {Chow},
  \citenamefont {Ryan}, \citenamefont {Rigetti}, \citenamefont {Poletto},
  \citenamefont {Ohki}, \citenamefont {Ketchen},\ and\ \citenamefont
  {Steffen}}]{Gambetta2012}%
  \BibitemOpen
  \bibfield  {author} {\bibinfo {author} {\bibfnamefont {J.~M.}\ \bibnamefont
  {Gambetta}}, \bibinfo {author} {\bibfnamefont {A.~D.}\ \bibnamefont
  {C\'{o}rcoles}}, \bibinfo {author} {\bibfnamefont {S.~T.}\ \bibnamefont
  {Merkel}}, \bibinfo {author} {\bibfnamefont {B.~R.}\ \bibnamefont {Johnson}},
  \bibinfo {author} {\bibfnamefont {J.~A.}\ \bibnamefont {Smolin}}, \bibinfo
  {author} {\bibfnamefont {J.~M.}\ \bibnamefont {Chow}}, \bibinfo {author}
  {\bibfnamefont {C.~A.}\ \bibnamefont {Ryan}}, \bibinfo {author}
  {\bibfnamefont {C.}~\bibnamefont {Rigetti}}, \bibinfo {author} {\bibfnamefont
  {S.}~\bibnamefont {Poletto}}, \bibinfo {author} {\bibfnamefont {T.~A.}\
  \bibnamefont {Ohki}}, \bibinfo {author} {\bibfnamefont {M.~B.}\ \bibnamefont
  {Ketchen}}, \ and\ \bibinfo {author} {\bibfnamefont {M.}~\bibnamefont
  {Steffen}},\ }\href {\doibase 10.1103/PhysRevLett.109.240504} {\bibfield
  {journal} {\bibinfo  {journal} {Phys. Rev. Lett.}\ }\textbf {\bibinfo
  {volume} {109}},\ \bibinfo {pages} {240504} (\bibinfo {year} {2012})},\
  \Eprint {http://arxiv.org/abs/1204.6308} {arXiv:1204.6308} \BibitemShut
  {NoStop}%
\bibitem [{\citenamefont {Fogarty}\ \emph {et~al.}(2015)\citenamefont
  {Fogarty}, \citenamefont {Veldhorst}, \citenamefont {Harper}, \citenamefont
  {Yang}, \citenamefont {Bartlett}, \citenamefont {Flammia},\ and\
  \citenamefont {Dzurak}}]{Fogarty2015}%
  \BibitemOpen
  \bibfield  {author} {\bibinfo {author} {\bibfnamefont {M.~A.}\ \bibnamefont
  {Fogarty}}, \bibinfo {author} {\bibfnamefont {M.}~\bibnamefont {Veldhorst}},
  \bibinfo {author} {\bibfnamefont {R.}~\bibnamefont {Harper}}, \bibinfo
  {author} {\bibfnamefont {C.~H.}\ \bibnamefont {Yang}}, \bibinfo {author}
  {\bibfnamefont {S.~D.}\ \bibnamefont {Bartlett}}, \bibinfo {author}
  {\bibfnamefont {S.~T.}\ \bibnamefont {Flammia}}, \ and\ \bibinfo {author}
  {\bibfnamefont {A.~S.}\ \bibnamefont {Dzurak}},\ }\href {\doibase
  10.1103/PhysRevA.92.022326} {\bibfield  {journal} {\bibinfo  {journal} {Phys.
  Rev. A}\ }\textbf {\bibinfo {volume} {92}},\ \bibinfo {pages} {022326}
  (\bibinfo {year} {2015})},\ \Eprint {http://arxiv.org/abs/1502.05119}
  {arXiv:1502.05119} \BibitemShut {NoStop}%
\bibitem [{\citenamefont {Ball}\ \emph {et~al.}(2015)\citenamefont {Ball},
  \citenamefont {Stace}, \citenamefont {Flammia},\ and\ \citenamefont
  {Biercuk}}]{Ball2015}%
  \BibitemOpen
  \bibfield  {author} {\bibinfo {author} {\bibfnamefont {H.}~\bibnamefont
  {Ball}}, \bibinfo {author} {\bibfnamefont {T.~M.}\ \bibnamefont {Stace}},
  \bibinfo {author} {\bibfnamefont {S.~T.}\ \bibnamefont {Flammia}}, \ and\
  \bibinfo {author} {\bibfnamefont {M.~J.}\ \bibnamefont {Biercuk}},\ }\href
  {\doibase 10.1103/physreva.93.022303} {\bibfield  {journal} {\bibinfo
  {journal} {Phys. Rev. A}\ }\textbf {\bibinfo {volume} {93}},\ \bibinfo
  {pages} {022303} (\bibinfo {year} {2015})},\ \Eprint
  {http://arxiv.org/abs/1504.05307} {arXiv:1504.05307} \BibitemShut {NoStop}%
\bibitem [{\citenamefont {Wallman}(2015)}]{Wallman2015b}%
  \BibitemOpen
  \bibfield  {author} {\bibinfo {author} {\bibfnamefont {J.~J.}\ \bibnamefont
  {Wallman}},\ }\href@noop {} {\  (\bibinfo {year} {2015})},\ \Eprint
  {http://arxiv.org/abs/1511.00727} {arXiv:1511.00727} \BibitemShut {NoStop}%
\bibitem [{\citenamefont {Watrous}(2011)}]{Watrous2011}%
  \BibitemOpen
  \bibfield  {author} {\bibinfo {author} {\bibfnamefont {J.}~\bibnamefont
  {Watrous}},\ }\href {https://cs.uwaterloo.ca/~watrous/LectureNotes.html}
  {\enquote {\bibinfo {title} {{CS 766 Theory of Quantum Information}},}\
  }\bibinfo {howpublished} {Available online at
  \url{https://cs.uwaterloo.ca/~watrous/LectureNotes.html}} (\bibinfo {year}
  {2011})\BibitemShut {NoStop}%
\bibitem [{\citenamefont {Boyd}\ and\ \citenamefont
  {Vandenberghe}(2004)}]{Boyd2004}%
  \BibitemOpen
  \bibfield  {author} {\bibinfo {author} {\bibfnamefont {S.}~\bibnamefont
  {Boyd}}\ and\ \bibinfo {author} {\bibfnamefont {L.}~\bibnamefont
  {Vandenberghe}},\ }\href@noop {} {\emph {\bibinfo {title} {{Convex
  Optimization}}}}\ (\bibinfo  {publisher} {Cambridge University Press},\
  \bibinfo {address} {New York},\ \bibinfo {year} {2004})\BibitemShut {NoStop}%
\bibitem [{\citenamefont {Vandenberghe}\ and\ \citenamefont
  {Boyd}(1996)}]{Vandenberghe1996}%
  \BibitemOpen
  \bibfield  {author} {\bibinfo {author} {\bibfnamefont {L.}~\bibnamefont
  {Vandenberghe}}\ and\ \bibinfo {author} {\bibfnamefont {S.}~\bibnamefont
  {Boyd}},\ }\href {\doibase 10.1137/1038003} {\bibfield  {journal} {\bibinfo
  {journal} {SIAM Rev.}\ }\textbf {\bibinfo {volume} {38}},\ \bibinfo {pages}
  {49} (\bibinfo {year} {1996})}\BibitemShut {NoStop}%
\bibitem [{\citenamefont {Barvinok}(2002)}]{Barvinok2002}%
  \BibitemOpen
  \bibfield  {author} {\bibinfo {author} {\bibfnamefont {A.}~\bibnamefont
  {Barvinok}},\ }\href@noop {} {\emph {\bibinfo {title} {A Course in
  Convexity}}},\ Vol.~\bibinfo {volume} {54}\ (\bibinfo  {publisher} {American
  Mathematical Society},\ \bibinfo {address} {Providence},\ \bibinfo {year}
  {2002})\BibitemShut {NoStop}%
\bibitem [{Note2()}]{Note2}%
  \BibitemOpen
  \bibinfo {note} {$\protect \mathcal {E} (\rho )$ and $\protect \mathaccentV
  {tilde}07E{\protect \mathcal {E}}(\rho )=U \protect \mathcal {E} (U^\dagger
  \rho U )U^\dagger $ have equal diamond distance and average error
  rate.}\BibitemShut {Stop}%
\bibitem [{\citenamefont {Kliesch}\ \emph {et~al.}(2015)\citenamefont
  {Kliesch}, \citenamefont {Kueng}, \citenamefont {Eisert},\ and\ \citenamefont
  {Gross}}]{Kliesch2015}%
  \BibitemOpen
  \bibfield  {author} {\bibinfo {author} {\bibfnamefont {M.}~\bibnamefont
  {Kliesch}}, \bibinfo {author} {\bibfnamefont {R.}~\bibnamefont {Kueng}},
  \bibinfo {author} {\bibfnamefont {J.}~\bibnamefont {Eisert}}, \ and\ \bibinfo
  {author} {\bibfnamefont {D.}~\bibnamefont {Gross}},\ }\href@noop {} {\
  (\bibinfo {year} {2015})},\ \Eprint {http://arxiv.org/abs/1511.01513}
  {arXiv:1511.01513} \BibitemShut {NoStop}%
\end{thebibliography}%

\clearpage
\onecolumngrid
\fontsize{11}{13.6}\selectfont 

\section{Supplementary material}

\subsection{Quantum states and operations}

A $d$-level quantum system is fully characterized by its is density operator $\rho$, which is a Hermitian, positive semidefinite $d \times d$ matrix obeying $\tr(\rho) = 1$.
A \emph{quantum operation} or \emph{channel} $\E$ is a completely positive linear map from density operators to density operators~\cite{Nielsen2010, Watrous2011}.

There are a number of representations of a completely positive operator, each of which is useful for different purposes. The most well known is the representation in terms of Kraus operators. These are a set of operators $\{K_i\}$ 
that encapsulate the channel's action via $\E (\rho) = \sum_i K_i \rho K_i^\dagger$.
Moreover, $\sum_i K_i^\dagger K_i \leq I$ holds, where $I$ is the identity matrix, and equality occurs when $\E$ is trace preserving. 

Other representations include the Liouville operator $L(\E) =  \sum_i \overline{K_i} \otimes K_i$ where $\otimes$ denotes the tensor product. The Liouville operator is also known as the transition matrix, or natural representation. It is a matrix that acts on the vector obtained by stacking the columns of $\rho$, which we denote $|\rho)$ as in~\cite{Wallman2014}, in the same way that $\E$ acts on the density operator $\rho$. That is $L(\E)|\rho) = |\E(\rho))$.

Lastly, we will have cause to use the Choi-Jamio\l{}kowski matrix of a quantum operation $\E$, $J(\E) = d (\I_A~\otimes~\E_B)(| \psi_{\mathrm{Bell}} \rangle \langle \psi_{\mathrm{Bell}} |)$. Here $\I$ is the identity channel and $| \psi_{\mathrm{Bell}} \rangle = \frac{1}{\sqrt{d}}\sum_{j=1}^d \ket{j}\otimes\ket{j}$ is the maximally entangled state between systems $A$ and $B$ (this definition differs by a factor of $d$ to that in~\cite{Wallman2014}, instead we use the definition found in~\cite{Magesan2012, Watrous2011} so as to be consistent with the semidefinite program in~\cite{Watrous2009}, which would otherwise require minor modification). It can be computed from the Kraus operators $\{K_i\}$ with the formula $J(\E) = \sum_i |K_i)(K_i|$ (where $(K_i| = \bar{|K_i)}^T$).

This representation is useful because, unlike the other representations mentioned here, $J(\E)$ is positive semidefinite for any completely positive quantum operation (the Kraus operators and Liouville operator need not even have a complete set of eigenvectors). 
			
We will be interested in relating the average infidelity $r(\E)$ to the diamond distance $D(\E)$ as defined in the main text in Eqs.~(\ref{eq:avgerrorrate}) and (\ref{eq:diamond}), respectively. (We will always be comparing a noise process to the identity channel, so we write the diamond distance with only one argument for brevity.) A useful formula is provided by the following relation which is a generalization of the main results in~\cite{Horodecki1999a, Nielsen2002} to completely positive maps that are not necessarily trace preserving.

\begin{proposition} 
\label{prop:infid_formula}
Let $\mathcal{E}$ be a completely positive (but not necessarily trace preserving) map with Liouville representation $L(\mathcal{E})$. Then
\begin{equation}
\label{eqn:infid_formula}
	F_\mathrm{avg}(\E) = \frac{\tr[L(\E)]+\tr[\E(I)]}{d(d+1)},
\end{equation}
where $F_\mathrm{avg} (\E) = 1-r(\E)$ is the average fidelity and $d$ is the system size.
\end{proposition}

Note that this formula covers the main results in~\cite{Horodecki1999a,Nielsen2002} as a special case. 
Indeed, any trace preserving map obeys $\tr \left( \E (I) \right) = d$ and Eq.~\eqref{eqn:infid_formula} reduces to~\cite{Horodecki1999a}[Proposition 1] and~\cite{Nielsen2002}[Equation (3)], respectively. 
For the scope of our work, such a generalization is very useful, since it will allow us to evaluate the fidelity of leakage processes averaged over qubit states.

\begin{proof}[Proof of  \autoref{prop:infid_formula}]

One way of proving the generalized formula \eqref{eqn:infid_formula} is to follow Nielsen's simplified proof steps ~\cite{Nielsen2002} of the original formula~\cite{Horodecki1999a} without assuming that $\mathcal{E}$ is trace preserving. At the core of this proof is the fact that the average fidelity
is invariant under twirling, i.e.\ $F_{\mathrm{avg}} \left( \E \right) = F_{\mathrm{avg}} \left( \E_T \right)$
for $\E_T (\rho) := \int \mathrm{d} U U^\dagger \E \left( U \rho U^\dagger \right) U^\dagger$.
Here $\mathrm{d}U$ denotes the unique unitarily invariant (Haar) measure over the unitary group $U(d)$ normalized to one ($\int \mathrm{d} U = 1$). 
The same is true for the r.h.s.\ of Eq.~\eqref{eqn:infid_formula}.
Indeed, suppose that $\E$ has Kraus representation $\E (\rho) = \sum_i K_i \rho K_i^\dagger$. 
Twirling it results in the map 
$\E_T (\rho) = \int \mathrm{d} U U^\dagger \sum_i (K_i U \rho U^\dagger K_i^\dagger) U$ whose Liouville representation obeys
\begin{align*}
\tr \left( L \left( \E_T \right) \right)
=& \tr \left( \int \mathrm{d}U \sum_i \bar{U}^\dagger \bar{K}_i \bar{U} \otimes U^\dagger K_i U \right)
= \int \mathrm{d} U \sum_i \tr \left( \bar{U}^\dagger \bar{K}_i \bar{U} \right) \tr \left( U^\dagger K_i U \right) \\
=& \sum_i \tr \left( \bar{K}_i \right) \tr \left( K_i \right) \int \mathrm{d} U = \tr \left( \sum_i \bar{K}_i \otimes K_i \right) = \tr \left( L \left( \E \right) \right).
\end{align*}
Also
\begin{equation*}
\tr \left( \E_T (I) \right) = \int \mathrm{d}U \tr \left( U^\dagger \E \left( U I U^\dagger \right) U \right)
= \tr \left( \E \left( I \right) \right)
\end{equation*}
which establishes twirl invariance of the r.h.s.\ of \eqref{eqn:infid_formula}.
As a result, it suffices to establish the claimed equality for twirled maps only.
However, due to Schur's Lemma, every twirl of a completely positive map is proportional to 
a depolarizing operation 
\begin{equation}
\E_T (\rho) =\mathcal{D}_{p,q} (\rho) :=  p \rho + q \tr (\rho) I\quad \forall \rho \label{eq:infid_aux1}
\end{equation}
with parameters $p,q$ that may depend upon the original map $\E$. 
Nielsen~\cite{Nielsen2002} established this by using the following elementary argument
based on the observation that any twirled channel obeys
\begin{equation}
V \E_T \left( \rho \right) V^\dagger = \E_T \left( V \rho V^\dagger \right) \quad \forall V \in U(d), \; \forall \rho \label{eqn:infid_aux1}
\end{equation}
which is readily established by direct computation.
Now let $ X = |x \rangle \langle x|$ be a rank one projector, set $X^\perp = I - X$ 
and let $V$ be an arbitrary unitary operator obeying $V X V^\dagger = X$. 
Inserting these particular choices into \eqref{eqn:infid_aux1}
reveals $\E_T ( X) = \E_T \left( V X V^\dagger \right) = V \E_T \left( X \right) V^\dagger$
which in turn implies $\E_T (X) = (p+q) X + q X^\perp = p X + q I$ for some $p,q \in \mathbb{R}$. 
A priori, the parameters $p,q$ may depend on the choice of $X$, but \eqref{eqn:infid_aux1} implies that they are actually the same for any choice of $X$. 
From this, Formula \eqref{eq:infid_aux1} is readily deduced, e.g.\ by inserting eigenvalue decompositions $\rho = \sum_{i=1}^d \lambda_i |x_i \rangle \langle x_i|$ of arbitrary density operators and exploiting linearity.

As a result, it suffices to establish Formula \eqref{eqn:infid_formula} exclusively for depolarizing maps $\mathcal{D}_{p,q}$ of the form \eqref{eq:infid_aux1} with parameters $p,q$. 
Noting that such a map has Liouville representation
$L \left( \mathcal{D}_{p,q} \right) = p I \otimes I + q d | \psi_{\mathrm{Bell}} \rangle \langle \psi_{\mathrm{Bell}}|$, where $| \psi_{\mathrm{Bell}} \rangle = \frac{1}{\sqrt{d}}\sum_{i=1}^d |i \rangle \otimes | i \rangle$ denotes a maximally entangled state, and calculating
\begin{align*}
F_{\mathrm{avg}} \left( \mathcal{D}_{p,q} \right)
=& p \int \mathrm{d} \psi \langle \psi | \psi \rangle \langle \psi | \psi \rangle + 
q \int \mathrm{d} \psi \tr \left( | \psi \rangle \langle \psi | \right) \langle \psi | I | \psi \rangle 
= p + q, \\
\tr \left( L \left( \mathcal{D}_{p,q} \right) \right)
=& p \tr \left( I \otimes I \right) + q d \tr \left( | \psi_{\mathrm{Bell}} \rangle \langle \psi_{\mathrm{Bell}}| \right) = d^2 p + d q, \\
\tr \left( \mathcal{D}_{p,q} \left( I \right) \right)
=& p \tr \left( I \right) + q \tr \left( I \right)^2 = d p + d^2 q
\end{align*}
reveals
\begin{equation*}
\tr \left( L \left( \mathcal{D}_{p,q} \right) \right) + \tr \left( \mathcal{D}_{p,q} \right)
= (d+1)d(p+q) = (d+1)d F_{\mathrm{avg}} \left( \mathcal{D}_{p,q} \right),
\end{equation*}
thus establishing the desired statement.
\end{proof}

\subsection{Semidefinite Programming}
\label{subsec:semidef}

It is possible to efficiently calculate the diamond norm of a linear operator through the use of a \emph{semidefinite program} if a full description of the channel is known~\cite{Watrous2009, Ben-Aroya2010, Watrous2013}.

A semidefinite program (SDP) is a form of mathematical optimization problem (specifically a convex optimization problem; see~\cite{Boyd2004, Vandenberghe1996} for a review). A mathematical optimization problem is very generally a specification of some objective function to be maximized (or minimized), subject to some constraints on allowed variables in the form of inequalities involving constraint functions. This can be stated in the form
\begin{center}
	\begin{tabular}{l l}
		Maximize: & $f_0(z)$ \\
		Subject to: & $f_i(z) \leq b_i$, \quad $i = 1,...,m$. \\ 
	\end{tabular}
\end{center}
where $f_0$ is the objective function, the $f_i$'s and $b_i$'s encode the constraint functions, and $z$ is the variable to be changed so as to maximize $f_0(z)$. Any value of $z$ which meets the constraints of the problem is called feasible. In some contexts these problem specifications are called programs.

A convex optimization problem is a mathematical optimization problem in which the set of all feasible points is a convex set and the objective function to be maximized is concave, i.e.\ it satisfies $f(\tau~x~+~(1-\tau) y) \geq \tau~f~(x)~+~(1-\tau)~f (y)$ for any $\tau \in [0,1]$ and feasible $x,y$. 
Note that minimising a convex function $f_0$ over a convex set also fits this framework, because it is equivalent to maximising $(-f_0)$ which is concave. 
Concave functions have many desirable properties that render convex optimization tasks easier than general optimization problems (e.g.\ concavity assures that any local maximum is also a global maximum)~\cite{Barvinok2002}.

Finally, a semidefinite program is a particular instance of a convex optimization problem where one aims to maximize a linear function (which is both concave and convex) over a convex subset of the cone of positive semidefinite matrices~\cite{Barvinok2002}. 
This cone induces a partial ordering on the space of all hermitian $d \times d$ matrices.
Concretely, we write $X \geq Y$ if and only if $X - Y$ is positive semidefinite.
With this notational convention, every SDP is of the form
\begin{equation}
	\begin{tabular}{l l}
		Maximize: & $\tr \left( C X \right)$ \\
		Subject to: & $ \Xi (X) \leq B$, \\
			& $ X \geq 0$\,.
	\end{tabular}
	\label{eqn:primal}
\end{equation}
and is specified by a triple $(\Xi,B,C)$: $B$ and $C$ are hermitian matrices (not necessarily of the same dimensions) and $\Xi$ is a linear map between these matrices spaces.
An SDP of the form \eqref{eqn:primal} is called a \emph{primal program}.
In a geometric sense, the problem here is to move as far along the direction of $C$ as possible, while remaining inside the convex region specified by the matrix inequalities~\cite{Boyd2004,Vandenberghe1996,Barvinok2002}.
A wide variety of problems can be cast in terms of semidefinite programs and efficient methods are known that can solve them. Thus, finding an expression for a problem in terms of a semidefinite program reduces it to one in which the solution is easily found numerically, and sometimes even analytically. 

Attached to every primal problem is another semidefinite program \eqref{eqn:primal}, known as its \emph{dual program}.
In a sense, it corresponds to a reverse problem and is given by
\begin{equation}
	\begin{tabular}{l l}
		Minimize: & $\tr \left( Z B \right)$ \\
		Subject to: & $\Xi^* (Z) \geq C$  \\
		& $Z \geq 0$, \\ 
	\end{tabular}
	\label{eqn:dual}
\end{equation}
which is again completely specified by the triple $(\Xi,C,B)$.  Here, $\Xi^*$ denotes the adjoint map of $\Xi$ with respect to the trace-inner product, i.e.\ the unique map obeying $\tr \left( \Xi^* (Z) X \right) = \tr \left( Z \; \Xi (X) \right)$ for all hermitian matrices $X$ and $Z$. 

Primal and dual SDP's are intimately related. In particular they have the property that any feasible value of the primal objective $\tr (C X )$ is less than or equal to any feasible value of the dual objective $\tr (Z B )$. Using the fact that positive semidefinite matrices $A,B,C \geq 0$ obey $\tr (A B) \leq \tr (A C)$ if and only if $B \leq C$ allows for an easy proof of this feature~\cite{Barvinok2002} via
\begin{equation*}
	\tr \left( C X \right) \leq \tr \left( \Xi^* (Z) X \right) = \tr \left( Z \; \Xi (X) \right)
	\leq \tr \left( Z B \right),
\end{equation*}
where we also have employed the constraints in \eqref{eqn:dual} and \eqref{eqn:primal}, respectively.
This result is known as \emph{weak duality}.  Typically an even stronger relation -- called \emph{strong duality} -- is true, namely that the optimum values of both problems coincide. 

Weak duality allows us to find an upper bound for the optimum value of \eqref{eqn:primal} in the form of any feasible value of \eqref{eqn:dual}. 
To be more explicit, if $Z$ is feasible, then $\tr (Z B)$ must be  larger than or equal to any feasible $\tr (C X)$. 
This in particular includes the maximal value $\tr (C X^\sharp)$ of \eqref{eqn:primal}.
However, since $\tr (C X^\sharp )$ is maximal, it is by definition larger than or equal to any feasible value of $\tr (C X)$. 
Consequently, the feasible values $\tr (C X )$ and $\tr (Z B)$ certify that the optimum primal value $\tr (C X^\sharp)$ is in a certain range. These bounds are said to be \emph{certificates}. 
Throughout this work, we will employ such certificates in order to find bounds for the diamond norm.
What is more, if we can find a pair of feasible points $X,Z$ that obey $\tr (C X ) = \tr (Z B)$, then weak duality dictates that we have analytically found the optimum value for the program. 
We will also appeal to this fact later.

\subsection{Semidefinite programs for the diamond distance}

Watrous has provided several characterisations of the diamond distance in terms of semidefinite programs~\cite{Watrous2009, Watrous2013}. We reproduce here a simplified version that can be used when the operator in question is a difference of quantum channels $\Delta=\mathcal{E}-\mathcal{F}$~\cite{Watrous2009}, as this will always be the case for us. Given this condition, the following pair of primal and dual SDP's has an optimal value of $D = \half \dnorm{\Delta}$:
\begin{center}
\begin{tabular}{cc}
\begin{minipage}{7.5cm}
	\begin{equation}
		\begin{tabular}{l l}
			\textbf{Primal problem}
			\vspace{0.5cm} \\
			Maximize: & $\avg{J(\Delta),W}$ \\
	 		Subject to: & $W \leq \rho \otimes I_d$, \\
			& $\tr(\rho) = 1$, \\
	 		& $W \in \mathrm{Pos}(A \otimes B)$, \\
	 		& $\rho \in \mathrm{Pos}(A)$. \\
	 	\end{tabular}
	 	\label{eq:diamond_primal}
	\end{equation}

\end{minipage}
& \hspace{1cm}
\begin{minipage}{7.5cm}
	\begin{equation}
		\begin{tabular}{l l}
	 		\textbf{Dual problem}
			\vspace{0.5cm}\\
	 		Minimize: & $\| \tr_B(Z) \|_\infty$ \\
	 		Subject to: & $Z \geq J(\Delta) $, \\
	 		& $Z \in \mathrm{Pos}(A\otimes B)$. \\
			& \quad \\
			& \quad
	 	\end{tabular}
	 	\label{eq:diamond_dual}
 	\end{equation}
\end{minipage}
\end{tabular}
\end{center}

Here $\avg{X,Y} = \tr(X^\dagger Y)$ is the Hilbert-Schmidt inner product of the matrices $X$ and $Y$, $\mathrm{Pos}(A \otimes B)$ denotes the cone of positive semidefinite operators acting on the system $A \otimes B$ and $\tr_B(X)$ is the partial trace of $X$ over system $B$, i.e.\ the subsystem of $X$ obtained when subsystem $B$ is discarded. 
Also, $\| X \|_\infty$ denotes the operator norm of $X$, which is the maximum eigenvalue of $X$ (if $X \geq 0$). 
Further information on these functions and spaces can be found in~\cite{Nielsen2010, Watrous2011}.

Note that, stated as it is, the primal problem is almost, but not quite, of the primal SDP form introduced in \eqref{eqn:primal}. 
However, some straightforward manipulations allow one to convert this problem into such a standard form.
Perhaps a bit surprisingly, the same is true for the dual problem which can also be recast as an instance of a dual SDP problem~\cite{Watrous2009}.

Finally, note that if $\Pi_+$ is the projector onto the positive eigenspace of $J(\Delta)$, then $\rho = \frac{1}{d}I$, $W = \frac{1}{d} \Pi_+$ are valid primal feasible values and $Z = \Pi_+ J(\Delta) \Pi_+$ is dual feasible. 
These feasible points were identified by Magesan, Gambetta, and Emerson~\cite{Magesan2012}, and inspired by their approach we will use similar constructions of primal and dual feasible points to get bounds on the diamond norm for various noise processes.

\subsection{Dephasing and calibration errors for a single qubit}

The channel described in the main text has Kraus operators $K_0=\sqrt{1-p}U(\delta)$ and $K_1=\sqrt{p}U(\delta)\sigma_z$, where $U(\delta)=\exp(-i\delta \sigma_z)$. Using the formula above for the average fidelity it is straightforward to show that
\begin{equation*}
	r_{\text{CD}} = \frac{2}{3}\bigl[p\cos(2\delta)+\sin^2\delta\bigr]\,.
\end{equation*}
Likewise evaluating the upper and lower bounds on $D_{\text{CD}}$ arising from the primal and dual feasible solutions of Ref.~\cite{Magesan2012} we find them to be equal and so obtain the result
\begin{equation*}
	 D_{\text{CD}}=\half \biggl| 1- (1-2p)e^{2i\delta} \biggr|\,.
\end{equation*}
A simple algebraic manipulation then shows the result claimed in the main text
\begin{equation*}
	D_{\text{CD}} = \sqrt{\frac{3}{2}r_{\text{CD}}-p(1-p)}.
\end{equation*}
			
\subsection{Thermal relaxation of a single qubit}

This one-qubit channel $\E_{\text{AD}}$ is characterized by 2 parameters $p,\gamma \in [0,1]$ and four Kraus operators \cite[Chapter 8.3.5]{Nielsen2010}
\begin{equation*}
	K_0 = \sqrt{p} \begin{pmatrix}
					1 & 0 \\
					0 & \sqrt{1-\gamma}
				\end{pmatrix}, 
	\quad
	K_1 = \sqrt{p} \begin{pmatrix}
					0 & \sqrt{\gamma} \\
					0 & 0
				\end{pmatrix}, 
	\quad
	K_2 = \sqrt{1-p} \begin{pmatrix}
					\sqrt{1-\gamma} & 0 \\
					0 & 1
				\end{pmatrix}, 
	\quad
	K_3 = \sqrt{1-p} \begin{pmatrix}
					0 & 0 \\
					\sqrt{\gamma} & 0
				\end{pmatrix}.
\end{equation*}
Repeating the procedure outlined in the previous subsection, we will use a refined dual feasible point to find a bound on the diamond distance in terms of the average fidelity. This feasible point improves over what can be obtained using the Magesan-Gambetta-Emerson feasible solution~\cite{Magesan2012}.

\begin{theorem} \label{thm:AD}
For the one-qubit amplitude damping channel defined above, the following relation is valid for any choice of parameters $p,\gamma \in [0,1]$:
\begin{equation*}
	D_{\emph{AD}}  \leq 3 r_{\emph{AD}} \max\{p,1-p\}.
\end{equation*}
\end{theorem}

\begin{proof}
We first compute the Choi-Jamio\l{}kowski matrix $J(\Delta)$ for $\Delta = \mathcal{I} - \mathcal{E}_{\text{AD}}$. In the basis $\ket{00}, \ket{01}, \ket{10}, \ket{11}$, this matrix is
\begin{align}
	J(\Delta) = 
\begin{pmatrix}
 (1-p) \gamma & 0 & 0 & 1-\sqrt{1-\gamma} \\
 0 & -(1-p) \gamma & 0 & 0 \\
 0 & 0 & -p \gamma  & 0 \\
 1-\sqrt{1-\gamma} & 0 & 0 & p \gamma  \\
\end{pmatrix}.
\end{align}

The middle block is already negative semidefinite and so our dual feasible point $Z$ can afford to have zero support on this subspace and still meet the constraints of Eq.~(\ref{eq:diamond_dual}). Let us therefore make the ansatz that 
\begin{align}\label{eq:firstZgad}
	Z = \left(
\begin{array}{cccc}
  x+y_0 & 0 & 0 & x \\
 0 & 0  & 0 & 0 \\
 0 & 0 & 0  & 0 \\
  x& 0 & 0 &  x+y_1 \\
\end{array}
\right) 
= 2x | \psi_{\mathrm{Bell}} \rangle \langle \psi_{\mathrm{Bell}}| + y_0  | 0 0 \rangle \langle 0 0 | +y_1 |1 1 \rangle \langle 1 1| 
\end{align}
where $x= \left(1-\sqrt{1-\gamma}+\gamma/2\right)/2$ and we will determine the parameters $y_0,y_1 \geq 0$. Such a choice of parameters assures that $Z$ is positive semidefinite. 

The only other constraint that must be respected is that $Z - J(\Delta)$ must be positive semidefinite. Let us define $x_-= \left(1-\sqrt{1-\gamma}-\gamma/2\right)/2\geq 0$. Here we have used the elementary relation $1-\sqrt{1-\gamma} \geq \frac{\gamma}{2}$ (which follows from concavity of the square root). Secondly we can define $|\psi_{-\mathrm{Bell}}\rangle=(|00\rangle-|11\rangle)/\sqrt{2}$. In terms of this we may write
\begin{align}\label{eq:firstZgad}
	Z -J(\Delta)
=& 2x_- | \psi_{-\mathrm{Bell}} \rangle \langle \psi_{-\mathrm{Bell}}| + [y_0-(1/2-p)\gamma]  | 0 0 \rangle \langle 0 0 | + [y_1+(1/2-p)\gamma]  |1 1 \rangle \langle 1 1| \nonumber \\ &+(1-p)\gamma |01\rangle\langle 01|+p\gamma |10\rangle\langle 10 |.
\end{align}
Accordingly, this difference is positive semidefinite, if both
\begin{equation*}
y_0-(1/2-p)\gamma \geq 0
\quad \textrm{and} \quad 
y_1+(1/2-p)\gamma \geq 0
\end{equation*}
hold. 
Setting $y_0 = \max \{ \gamma/2-p\gamma,0\}$ and $y_1 = \max \{ 0, p\gamma-\gamma/2 \}$ satisfies the requirements. The two cases correspond to $p\leq 1/2$ and $p\geq 1/2$ respectively.  
Such a choice of parameters assures that $Z$ is a valid feasible point of the dual SDP \eqref{eq:diamond_dual} of the channel's diamond distance. Its objective function value amounts to
\begin{align*}
\| \mathrm{tr}_B \left( Z \right) \|_\infty
=& \| 2x\; \mathrm{tr}_B \left( | \psi_{\mathrm{Bell}} \rangle \langle \psi_{\mathrm{Bell}}| \right)
+ \mathrm{tr}_B \left(y_0 |00 \rangle \langle 0 0 | +y_1 |1 1 \rangle \langle 1 1| \right) \|_\infty  \\
=& \max\{ x+y_0,x+y_1\} \\ 
=& (1-\sqrt{1-\gamma} +\gamma/2 )/2+ \gamma|1-2p|/2 \\
\leq & (1-\sqrt{1-\gamma} +\gamma/2 )(1+|1-2p|)/2 \\
=& (1-\sqrt{1-\gamma} +\gamma/2 ) \max\{p,1-p\}\,.
\end{align*}
The inequality arises because $1-\sqrt{1-\gamma} \geq \gamma/2$ as noted above.

Using the formula of Eq.~(\ref{eqn:infid_formula}), one easily obtains $r_{\text{AD}} = \frac{1}{3} \bigl(1-\sqrt{1-\gamma} + \frac{\gamma}{2} \bigr) $. 
From this we may conclude
\begin{equation*}
D_{\mathrm{AD}} = \frac{1}{2} \| \Delta_{\mathrm{AD}} \|_\diamond \leq \| \mathrm{tr}_B (Z) \|_\infty \leq (1-\sqrt{1-\gamma} +\gamma/2 )\max\{p,1-p\}
= 3 r_{\mathrm{AD}}\max\{p,1-p\}.
\end{equation*}
This is the inequality that was to be proven. 
\end{proof}

\subsection{Incoherent leakage errors}

Our model of incoherent leakage errors for a single qubit is typical of a physical leakage process that may occur. We assume that the qubit state $|1\rangle$ can relax to a leakage state $|l\rangle$. We specify the noise process in terms of a leakage probability $p$ and Kraus operators
\begin{equation*}
K_0 = |0 \rangle \langle 0| + \sqrt{1-p} |1 \rangle \langle 1| +|l\rangle \langle l|, 
\quad
K_1 = \sqrt{ p} |l \rangle \langle 1 | .
\end{equation*}

To compute the average fidelity over initial qubit states we note that this average fidelity is unchanged if we replace the noise process with the a noise map where the Kraus operators are $\Pi_q K_i \Pi_q$ and $\Pi_q= |0 \rangle \langle 0| +  |1 \rangle \langle 1|$ is the projector on the qubit subspace. The resulting process maps the qubit subspace to the qubit subspace and is completely positive but not trace preserving. We can thus evaluate the average fidelity using \autoref{prop:infid_formula} which is valid for non-trace-preserving maps. Given this we find $r_{\text{IL}}= [4-(1+\sqrt{1-p})^2+p]/6=[1-\sqrt{1-p}+p]/3$. 

Note that if the average fidelity is computed over the full three-level space, the answer is slightly different and corresponds to $[1-\sqrt{1-p}+p/4]/3$. Using this alternate characterization of average error rate gives only minor quantitative and no qualitative changes to our conclusions. We therefore choose the average only over the qubit space as a more physically motivated quantity. 

To bound the diamond norm error we modify the dual feasible solution that worked for the thermal relaxation process above.
The Choi matrix of the channel difference is
\begin{equation*}
J \left( \Delta_{\mathrm{IL}} \right) 
= - p |1 1 \rangle \langle 1 1 | + p |1 l \rangle \langle 1 l | 
 + \left(\sqrt{1-p} - 1 \right) \left( |00 \rangle \langle 11 | + |11 \rangle \langle 00 | +  |ll \rangle \langle 11 | + |11 \rangle \langle ll | \right).
\end{equation*}

We choose
\begin{equation*}
Z = \left( 1 - \sqrt{1-p} \right)\left( |0 0 \rangle \langle 0 0 | + |l l \rangle \langle l l | + |00 \rangle \langle ll | + |ll \rangle \langle 00 |\right)+ p  |1 l \rangle \langle 1 l | .
\end{equation*}
as dual feasible point.
It is clear that $Z\geq 0$ and the second feasibility condition follows from
\begin{align*}
	Z - J \left( \Delta_{\mathrm{IL}} \right)
	=& 3 \left( 1- \sqrt{1-p} \right) | \psi_{\mathrm{Bell}} \rangle \langle \psi_{\mathrm{Bell}} | 
	+ \left[ p - \left( 1 - \sqrt{1-p} \right) \right] |11 \rangle \langle 11|,
\end{align*}
where here $|\psi_{\mathrm{Bell}} \rangle :=  \sum_{i=1}^3( |i \rangle \otimes | i \rangle)/\sqrt{3}$.
A routine calculation verifies that the coefficient in front of $|11 \rangle \langle 11|$ is nonnegative for any $p \in [0,1]$ and $Z - J \left( \Delta_{\mathrm{IL}} \right)$ is thus positive semidefinite.
Inserting $Z$ into the dual problem's objective function \eqref{eq:diamond_dual} yields
\begin{equation}
	D_{\text{IL}}
	\leq \left\| \tr_B \left( Z \right) \right\|_\infty 
	= \left\|  \left( 1 -\sqrt{1-p} \right)\left( |l \rangle \langle l | +  |0 \rangle \langle 0 |\right)+  p|1 \rangle \langle 1 | \right\|_\infty 
	=  p \leq 2r_{\text{IL}} . \label{eq:ILone}
\end{equation}
This is the inequality that we wished to show. (The final inequality follows because $p=(p+2p)/3\leq 2(1-\sqrt{1-p}+p)/3 =2r_{IL} $ since $p/2\leq 1-\sqrt{1-p}$ as noted above.)

We may also consider the following alternative model for incoherent leakage in $d$-dimensional quantum systems:
\begin{equation*}
\mathcal{E}_{\mathrm{IL}}(\rho) = p P \rho P + (1-p) \rho,
\end{equation*}
with $p \in [0,1]$ and $P$ is a rank-deficient orthogonal projection (i.e.\ $P \geq 0$, $P^2=P$ and $1 \leq \mathrm{tr}(P) \leq d-1$). For single qubits ($d=2$), $P$ necessarily coincides with a pure quantum state and we recover the incoherent leakage model examined in \cite[Eq.~(25)]{Wallman2015a}. This channel model has the advantage that we can exactly determine its diamond distance:
\begin{equation}
D_{\mathrm{IL}} = p. \label{eq:incoherent_leakage_distance}
\end{equation}

The related computations greatly simplifies if we exploit unitary invariance of both diamond distance and average error rate
\footnote{
$\mathcal{E} (\rho)$ and $\tilde{\mathcal{E}}(\rho)=U \mathcal{E} (U^\dagger \rho U )U^\dagger$ have equal diamond distance and average error rate.}. 
This unitary invariance allows us to w.l.o.g. assume that $P$ is diagonal in the computational basis: $P = \sum_{k=1}^{\mathrm{rank}(P)} |k \rangle \langle k|$. The Choi matrix of $\Delta_{\mathrm{IL}}=\mathcal{I}-\mathcal{E}_{\mathrm{IL}}$ then amounts to
\begin{equation*}
J \left( \Delta_{\mathrm{IL}} \right) = d \mathcal{I} \otimes \left( \mathcal{I}-\mathcal{E}_{\mathrm{IL}} \right) \left( | \psi_{\mathrm{Bell}} \rangle \langle \psi_{\mathrm{Bell}}| \right)
= p d\left( | \psi_{\mathrm{Bell}} \rangle \langle \psi_{\mathrm{Bell}}| - | \psi_P \rangle \langle \psi_P | \right),
\end{equation*}
where $| \psi_P \rangle = \frac{1}{\sqrt{d}} \sum_{k=1}^d |k k \rangle = \frac{1}{\sqrt{d}} \sum_{k=1}^{\mathrm{rank}(P)} |k k \rangle$. 
In order to obtain an upper bound, we choose the following feasible point of the diamond distance's dual SDP: $Z = pd | \psi_{\mathrm{Bell}} \rangle \langle \psi_{\mathrm{Bell}}|$. 
Clearly, this matrix is a feasible point, because $Z \geq 0$ and $Z - J (\Delta_{\mathrm{IL}})= pd | \psi_P \rangle \langle \psi_P| \geq 0$. It's corresponding objective function value amounts to
\begin{equation*}
\| \mathrm{tr}_B (Z) \|_\infty = pd \| \mathrm{tr}_B \left( | \psi_{\mathrm{Bell}} \rangle \langle \psi_{\mathrm{Bell}}| \right) \|_\infty
= pd \| \frac{1}{d} \mathbb{I} \|_\infty = p,
\end{equation*}
which serves as our upper bound on $D_{\mathrm{IL}}$. 

For a lower bound, we turn to the primal SDP of the diamond distance. We set $\rho = |d \rangle \langle d|$ and $W=|dd \rangle \langle dd|$ which is a feasible pair of primal variables ($W \leq \rho \otimes \mathbb{I}$, $\mathrm{tr}(\rho)=1$ and $W,\rho \geq 0$). 
Evaluating the primal objective function at this point results in
\begin{align*}
\left( J (\Delta_{\mathrm{IL}}), W \right) 
=& dp \left| \langle d d | \psi_{\mathrm{Bell}} \rangle \right|^2 - pd \left| \langle d d | \psi_P \rangle \right|^2
= p.
\end{align*}
Note that this lower bound on $D_{\mathrm{IL}}$ coincides with the upper bound established below. Weak duality allows us to conclude \eqref{eq:incoherent_leakage_distance}.

Finally, the average error rate of $\mathcal{E}_{\mathrm{IL}}$ can be readily computed via Formula~\eqref{eqn:infid_formula} and amounts to
\begin{equation*}
r_{\mathrm{IL}} = p \left( 1- \frac{ \mathrm{tr}(P) (\mathrm{tr}(P)+1)}{(d+1)d} \right)
\in \left[ \frac{2p}{d+1}, p \left( 1-\frac{2}{d(d+1)}\right) \right].
\end{equation*}
The upper bound is saturated for rank-one projectors $P$, while the lower bound is achieved for projectors with $\mathrm{rank}(P)=d-1$. Comparing this to $D_{\mathrm{IL}}=p$ reveals
\begin{equation*}
D_{\mathrm{IL}} = \left( 1- \frac{ \mathrm{tr}(P)(\mathrm{tr}(P)+1)}{(d+1)d} \right) r_{\mathrm{IL}} \leq \frac{d+1}{2} r_{\mathrm{IL}}.
\end{equation*}
The upper bound provided here is tight for $(d-1)$-dimensional projections and becomes increasingle loose for more rank-deficient ones. For single qubits ($d=2$), however, the upper bound is tight and we obtain
$
D_{\mathrm{IL}} = \frac{3}{2} r_{\mathrm{IL}}.
$. 
Finally, choosing $d=3$ and $\mathrm{rank}(P)=2$ mimics the dimensionalities ocurring in our previous model for incoherent leakage. For such a choice, we obtain
\begin{equation*}
D_{\mathrm{IL}} = 2 r_{\mathrm{IL}},
\end{equation*}
which agrees with \eqref{eq:ILone}, but is slightly stronger.

\subsection{Coherent leakage errors}

The coherent leakage process that we consider is a unitary error process 
\begin{equation}
U (\delta)=\exp [-i\delta (|1\rangle \langle l| +|l\rangle \langle 1|)] = |0\rangle \langle 0| + \cos (\delta) ( |1\rangle \langle 1| + |l\rangle \langle l| )-i \sin (\delta) (|1\rangle \langle l| +|l\rangle \langle 1|),
\label{eq:coherent_leakage}
\end{equation}
where $\delta \in [-\pi,\pi ]$ mediates the error strength.
We can derive the average-case error using the same trick as above of projecting onto the qubit subspace. Note that $\Pi_q U \Pi_q =  |0\rangle \langle 0| + \cos (\delta)  |1\rangle \langle 1| $. As a result we find $r_{\text{CL}} = [2-\cos\delta -\cos^2\delta]/3$. Unlike the incoherent case, the average error rate here is by coincidence the same if we compute it in the projected space or in the three-level space.

On the other hand the computation of the diamond norm distance is more straightforward for unitary error models such as this, since the optimization over input states entangled with an ancilla in the definition is not required. More details of the computation of the diamond norm distance for general unitary errors are given in the following subsection. The result of \autoref{cor:CL} is that $D_{\text{CL}}=|\sin\delta|$. 

To relate worst-case and average-case error, we employ the relation 
$4\sin^2(\delta/2)\geq \sin^2\delta$ which assures
\begin{equation*}
r_{\text{CL}}=(1-\cos\delta)/3+(1-\cos^2\delta)/3= 2\sin^2(\delta/2)/3+\sin^2(\delta)/3\geq \sin^2(\delta)/2 = D^2_{\text{CL}}/2.
\end{equation*}
On the other hand we can place a lower bound on the diamond norm distance. To tighten it, we will consider the case of moderately small error with $\delta \in [-\pi/2,\pi/2]$.
This assures $\cos^2\delta \leq \cos\delta$ and we obtain
\begin{equation*}
r_{\text{CL}}=(1-\cos\delta)/3+(1-\cos^2\delta)/3\leq 2[1- \cos^2(\delta)]/3=2 \sin^2(\delta)/3= 2 D^2_{\text{CL}}/3.
\end{equation*}
So for the restricted range of $\delta \in [-\pi/2,\pi/2]$ we have 
\begin{equation*}
\sqrt{3r_{\text{CL}}/2} \leq D_{\text{CL}} \leq \sqrt{2r_{\text{CL}}}
\end{equation*}
which is the inequality we intended to show and demonstrates that  the diamond norm distance scales with $\sqrt{r_{\text{CL}}}$.

\subsection{Unitary errors}
\label{sec:unitary}

In this section we do not restrict ourselves to qubits anymore and consider $d$-dimensional unitary channels, i.e.\  
\begin{equation*}
	\rho  \mapsto U \rho U^\dagger
\end{equation*}
where $U: \mathbb{C}^d \to \mathbb{C}^d$ is a unitary matrix ($U U^\dagger = U^\dagger U = I$). 
As we will show now, all channels of this form admit the unfavorable ``square root'' behavior where the worst-case error is roughly equal to the square root of the average case error.
We summarize our results as follows.

\begin{theorem} \label{thm:U}
Fix a dimension $d$ and let $\mathcal{E}_{\mathrm{U}}$ be a unitary channel.
Then 
\begin{equation}
	 \sqrt{\frac{d+1}{d}} \sqrt{r_{\mathrm{U}}} \leq  D_ \mathrm{U}
	\leq \sqrt{(d+1)d}\sqrt{r_{\mathrm{U}}}.\label{eq:U}
\end{equation}
Moreover, for single-qubit unitary channels, the lower bound holds with equality, i.e.\
$D_ \mathrm{U}= \sqrt{3 r_{\mathrm{U}}/2}$.
\end{theorem}

While the lower bound in \eqref{eq:U} is tight, we do not know if the dimensional dependence in the upper bound can be further improved and leave this for future work.

\begin{proof}[Proof of \autoref{thm:U}]
Every unitary matrix $U$ is normal and as such has an eigenvalue decomposition
\begin{equation*}
	U = \sum_{k=1}^d \mathrm{e}^{i \delta_k} |k \rangle \langle k|,
\end{equation*}
with eigenvalues $\mathrm{e}^{i \delta_k}$ on the complex unit circle and an orthogonal eigenbasis $\left\{ | k \rangle \right\}_{k=1}^d$ of $\mathbb{C}^d$. 
It greatly facilitates our work if we define the maximally entangled state $| \psi_{\mathrm{Bell}} \rangle = \frac{1}{\sqrt{d}} \sum_{k=1}^d |k,k \rangle$ with respect to this eigenbasis. 
With such a choice, the channel's Choi matrix simply corresponds to
\begin{equation*}
J \left( \mathcal{E}_{\mathrm{U}} \right) = d (\mathcal{E}_{\mathrm{U}} \otimes \mathcal{I}) \left( | \psi_{\mathrm{Bell}} \rangle \langle \psi_{\mathrm{Bell}} | \right)
= d \left( U \otimes \mathbbm{1} \right) | \psi_{\mathrm{Bell}} \rangle \langle \psi_{\mathrm{Bell}} |  \left( U^\dagger \otimes \mathbbm{1} \right)
= d | \phi_{\mathrm{U}} \rangle \langle \phi_{\mathrm{U}} |,
\end{equation*}
where $| \phi_{\mathrm{U}} \rangle = \frac{1}{\sqrt{d}} \sum_{k=1}^d \mathrm{e}^{i \delta_k} |k k \rangle$ is again a maximally entangled state.
The channel's average error rate then corresponds to
\begin{equation}
r _{\mathrm{U}} 
= \frac{d - \langle \psi_{\mathrm{Bell}} | J \left( \mathcal{E}_{\mathrm{U}} \right) | \psi_{\mathrm{Bell}} \rangle}{d+1}
= \frac{d -d \left| \langle \psi_{\mathrm{Bell}} | \phi_{\mathrm{U}} \rangle \right|^2 }{d+1}
= \frac{d^2 - \left| \sum_{k=1}^d \mathrm{e}^{i \delta_k} \right|^2}{d(d+1)}. \label{eq:U_r}
\end{equation}
For the upper bound in \eqref{eq:U}, we use the fact that the Choi matrix of the channel difference $\Delta_{\mathrm{U}} = \mathcal{E}_{\mathrm{U}} - \mathcal{I}$ assumes the form
\begin{equation*}
J \left( \Delta_{\mathrm{U}} \right) = d \left( \left( U \otimes \mathbbm{1} \right) | \psi_{\mathrm{Bell}} \rangle \langle \psi_{\mathrm{Bell}} | \left( U^\dagger \otimes \mathbbm{1} \right) - | \psi_{\mathrm{Bell}} \rangle \langle \psi_{\mathrm{Bell}} | \right)
\end{equation*}
which is proportional to the difference of two rank-one projectors. Such a matrix has two non-zero eigenvalues
\begin{equation*}
\lambda_{\pm} =  \pm d \sqrt{ 1 - \left| \langle \psi_{\mathrm{Bell}} | \phi_{\mathrm{U}} \rangle \right|^2} = \pm \sqrt{(d+1)d} \sqrt{ r_{\mathrm{U}}}
\end{equation*}
and corresponding normalized eigenvectors $|v_+ \rangle, |v_- \rangle \in \mathbb{C}^{d^2}$ -- see e.g.~\cite[Example 2.3]{Watrous2011}. 
Setting $Z = \lambda_+ | v_+ \rangle \langle v_+| \geq 0 $ yields a valid dual feasible point for the diamond norm's dual SDP \eqref{eq:diamond_dual} and inserting it into the program's objective function reveals
\begin{equation*}
D _{U}  
\leq \left\| \tr_B \left( Z \right) \right\|_\infty
\leq \| \tr_B \left( Z \right) \|_1 = \tr \left( Z \right)= \lambda_+ \langle v_+ | v_+ \rangle = \lambda_+ 
= \sqrt{(d+1)d} \sqrt{r _{\mathrm{U}}},
\end{equation*}
as claimed. Here we have made use of the basic norm inequality $\| \cdot \|_\infty \leq \| \cdot \|_1$ and the fact that the partial trace preserves positive semidefiniteness which in turn assures $\| \tr_Y (Z) \|_1 = \tr \left( \tr_Y (Z) \right) = \tr (Z)$.

For the lower bound, we use the fact that for the difference of two unitary channels, diamond norm and induced trace norm coincide \cite[Theorem 20.7]{Watrous2011}. This in turn assures
\begin{equation}
D_{\mathrm{U}} = \frac{1}{2} \left\| \mathcal{E}_{\mathrm{U}}- \mathcal{I}  \right\|_{1 \to 1} 
= \frac{1}{2} \max_{\| x \|_{\ell_2} = 1} \left\| U |x \rangle \langle x| U^\dagger - |x \rangle \langle x| \right\|_1
= \max_{\|x\|_{\ell_2} = 1} \sqrt{1 - \left| \langle x|U|x \rangle \right|^2}, \label{eq:U_diamond1}
\end{equation}
where the last simplification once more exploits that the matrix of interest is a difference of two rank-one projectors.
Choosing the particular vector $\tilde{x} =  \sum_{k=1}^n |k \rangle/\sqrt{d} $ allows us to also conclude
\begin{equation}
D _{\mathrm{U}}
\geq \sqrt{ 1 - \left| \langle \tilde{x} | U |\tilde{x} \rangle \right|^2 } 
= \frac{1}{d} \sqrt{ d^2 - \left| \sum_{k=1}^d \mathrm{e}^{i \delta_k} \right|^2}
= \sqrt{\frac{d+1}{d}} \sqrt{ r _{\mathrm{U}}}, \label{eq:U_aux1}
\end{equation}
which is the lower bound presented in \eqref{eq:U}. 

For single-qubit unitary channels this argument can be substantially strengthened:
 in fact the inequality sign in \eqref{eq:U_aux1} can be replaced with actual equality.
To see this, we first note that any unitary channel $\mathcal{E}_{\mathrm{U}}$  is invariant under a global phase change $U \mapsto \mathrm{e}^{i \phi} U$ in the defining unitary matrix. 
For two-dimensional unitaries, this gauge freedom assures that we can w.l.o.g. assume that $U$ is of the form $\mathrm{e}^{i \delta} |0 \rangle \langle 0| + \mathrm{e}^{-i \delta} |1 \rangle \langle 1|$ with $\delta \in [-\pi,\pi]$. This in turn assures that any vector  $x = x_1 |0 \rangle + x_2 |1 \rangle \in \mathbb{C}^2$ obeys 
\begin{equation*}
\left| \langle x, U x \rangle \right|^2 
= \left| \mathrm{e}^{i \delta} |x_1|^2 + \mathrm{e}^{-i \delta} |x_2|^2 \right|^2  
= |x_1 |^4 + 2 \cos \left( 2 \delta \right) |x_1 |^2 |x_2|^2 + |x_2|^4.
\end{equation*}
Clearly, this function is ignorant towards individual phases of $x_1,x_2$ and when attempting to minimize it, we may focus on real coefficients only. 
Taking into account normalization allows us to restrict $x_1$ to the interval $[0,1]$ and setting $x_2^2=1-x_1^2$. Doing so reveals
\begin{equation}
\min_{\|x \|_{\ell_2} = 1} \left| \langle x, U x \rangle \right|^2
= \min_{x_1 \in [0,1]} \left( x_1^4 + 2 \cos (2 \delta)x_1^2 (1-x_1^2) + \left( 1 - x_1^2 \right)^2 \right)  
 = \min_{x_1 \in [0,1]} \left( 4 \sin^2 (\delta) \left( x_1^4 - x_1^2 \right) + 1 \right) \label{eq:double_well}
\end{equation}
and maximizing the expression on the r.h.s. of \eqref{eq:U_diamond1} is therefore equivalent to finding the minimum of the particularly simple double-well potential \eqref{eq:double_well}.
The minimal value of the latter is achieved for $x_1 = 1/\sqrt{2}$, which in turn assures that the vector $\tilde{x} = \left( |0 \rangle + |1 \rangle \right)/\sqrt{2}$ in fact
 minimizes $\left| \langle x, U x \rangle \right|^2$ and 
-- as claimed --  the inequality sign in \eqref{eq:U_aux1} can be replaced by equality.
\end{proof}

Similar techniques can be employed to exactly characterize the diamond distance 
of single qubit coherent leakage, as it was introduced in the previous subsection.

\begin{corollary}[Diamond distance of coherent leakage]\label{cor:CL}
Consider the three-level coherent leakage channel $U (\delta)$ with $\delta \in [-\pi,\pi]$ introduced in \eqref{eq:coherent_leakage}. 
Then, its diamond distance amounts to
$
D_{\mathrm{CL}} = \left| \sin (\delta) \right|.
$
\end{corollary}

\begin{proof}
We start by noting that $U(\delta)$ as introduced in \eqref{eq:coherent_leakage}
admits an eigenvalue decomposition of the form 
$U (\delta) = \left( |v_0 \rangle \langle v_0| + \mathrm{e}^{i \delta} |v_+ \rangle \langle v_+ |
+ \mathrm{e}^{-i \delta} |v_- \rangle \langle v_-| \right)$,
where $|v_0 \rangle, |v_+ \rangle, |v_- \rangle$ form an orthonormal basis of $\mathbb{C}^3$. 
Since this channel is unitary, we can employ the particularly simple formula \eqref{eq:U_diamond1}
to calculate it's diamond distance:
\begin{equation}
D_{\mathrm{CL}} = \max_{\|x\|_{\ell_2}=1} \sqrt{1 - \left| \langle x | U (\delta) |x \rangle \right|^2}
\label{eq:CLaux1}
\end{equation}
Now note that for any vector $x = x_1 |v_0 \rangle + x_2 |v_+ \rangle + x_3 |v_- \rangle$ 
(represented with respect to the eigenbasis of $U(\delta)$), we have
\begin{equation*}
\left| \langle x | U (\delta) |x \rangle \right|^2
= \left| \left| x_1 \right|^2 + \left| x_2 \right|^2 \mathrm{e}^{i \delta} + \left| x_3 \right|^2 \mathrm{e}^{-i \delta } \right|^2.
\end{equation*}
An analysis similar to the one presented at the end of the proof of \autoref{thm:U}
reveals that such an expression is minimal for $x_1 =0$ and $\left| x_2 \right|^2 = \left|x_3 \right|^2 = 1/2$. Inserting such an optimal vector into \eqref{eq:CLaux1}
implies
\begin{equation*}
D_{\mathrm{CL}} = \max_{\| x \|_{\ell_2}=1} \sqrt{1 - \left| \langle x| U (\delta) | x \rangle \right|^2}
= \sqrt{1 - \cos^2(\delta)} = \left| \sin (\delta) \right|,
\end{equation*}
as claimed.
\end{proof}

\subsection{The unitarity and average error rate for two-qubit processes}

We now consider the noise process on two qubits in the main text, generated by $\mathrm{e}^{iH_{\text{CD2}}}$ where $H_{\text{CD2}} = \delta_1 \sigma_z^{(1)} + \delta_2 \sigma_z^{(2)} + \epsilon \sigma_z^{(1)}\sigma_z^{(2)}$. Because the unitarity and average error rate can be computed directly, without the need of analyzing a semidefinite program, we can simply use the formulas \eqref{eqn:infid_formula} and \eqref{eqn:unitarity_2norm} (below) and do a direct computation. The average error rate is given by
\begin{align*}
	r_{\text{CD2}} = \frac{1}{10} \bigl[ & 4 (2 p-1) \cos (2 \delta ) \cos (2 \epsilon ) -(1-2 p)^2 \cos (4 \delta )+4p(1-p)+5\bigr] \,,
\end{align*}
and the unitarity is given by
\begin{align*}
	u_{\text{CD2}} = \tfrac{1}{15}\bigl([8p(1-p)-4]^2-1\bigr) \,.
\end{align*}
Here for simplicity we have choosen $\delta_1 = \delta_2 =\delta$. This computation is routine, so we omit the details.

\subsection{The unitarity as a witness for unfavorable scaling}

The key message of this work is that the diamond distance $D(\E)$ of an error channel $\E$ may be proportional to the square root of its average error rate $r (\E)$. 
This is undesirable, since it underlines that $D(\E)$ -- which is the crucial number for fault tolerance -- may be orders of magnitude larger than $r (\E)$ -- a quantity that is routinely estimated via randomized benchmarking techniques. 
However, in our case studies we have found that for many channels this worst case behavior does not occur and there is a linear relationship $D(\E) = \mathcal{O} \left( r (\E) \right)$. 
In this section, we provide a necessary and sufficient criterion for such a desirable relationship. It is based on the \emph{unitarity}, a scalar that was introduced in~\cite{Wallman2015} 
and quantifies the coherence (i.e.\ the ``unitarity'') of a given noise channel $\E$.
To properly define it, we associate $\E$ with a reduced map $\E'$ that obeys
$\E' (I) = 0$ as well as
$\E' (X) = \E (X) - \frac{ \tr \left( \E (X) \right)}{\sqrt{d}} I$ for every traceless $X$. 
We define the unitarity of $\E$ to be the following averaged quantity of the reduced map $\E'$:
\begin{equation}
\label{eq:unitarity}
	u \left( \E \right) := \frac{d}{d-1} \int \mathrm{d} \psi\, \tr \Bigl( \E' ( | \psi \rangle \langle \psi | )^\dagger \E' ( | \psi \rangle \langle \psi | ) \Bigr).
\end{equation}
Defined that way, the unitarity obeys $u (\I) = 1$ and its definition in terms of $\E'$ makes it sensitive towards possible
non-unital and trace decreasing features of $\E$. In particular, it is also \emph{insensitive} to unitary rotations, in the sense that if $\mathcal{U}$ and $\mathcal{V}$ are unitary quantum channels, then 
$u \left( \mathcal{U} \E \mathcal{V} \right) = u \left (\E \right)$ 
holds true for any quantum channel $\E$.
As a result, the unitarity is independent of unitary pre- and post-rotations on the noise~\cite{Wallman2015}. 
The unitarity boasts many other desirable properties and -- perhaps most importantly -- can be efficiently estimated via a modified randomized benchmarking experiment~\cite{Wallman2015}. 
Moreover, it is related to the average error rate by means of the following inequality.

\begin{proposition} \label{prop:unitarity_vs_rate}
Let $\E$ be a not necessarily trace preserving quantum operation obeying  $\tr \bigl(\E ( I ) \bigr) \leq \tr (I)$.
Then the unitarity and average error rate of $\E$ obey
\begin{equation}
u( \E ) \geq \biggl( 1 - \frac{d r( \E )}{d-1} \biggr)^2, \label{eqn:unitarity1}
\end{equation}
where $d$ denotes the dimension of the system.
\end{proposition}

This is a slightly more general version of the inequality in~\cite{Wallman2015}[Proposition 8] and we provide a new proof based on fundamental Schatten-norm inequalities below.
For now, we content ourselves with stating the main result of this section:
for a large family of error channels, nearly saturating the bound \eqref{eqn:unitarity1}
is a necessary and sufficient condition for the desirable scaling relation $D(\E) = \mathcal{O} (r (\E))$. 

\begin{theorem} \label{thm:unitarity_witness}
Let $\E$ be an arbitrary unital and trace-preserving channel.
Then the diamond distance $D(\E)$ scales linearly in the average error rate $r = r (\E)$, if and only if 
the bound \eqref{eqn:unitarity1}
is saturated up to second order in $r(\E)$, i.e.\ 
\begin{equation}
u (\E) = \left( 1 - \frac{d r}{d-1} \right)^2 + \mathcal{O} \left( r^2 \right).
\label{eqn:unitarity_scaling}
\end{equation}
\end{theorem}

Since both $r (\E)$ and $u (\E)$ can be efficiently estimated in actual experiments, 
\autoref{thm:unitarity_witness} provides an efficient means to check
whether or not $D(\E)$ and $r(\E)$ are of the same magnitude. 
It immediately follows from the following technical result. 

\begin{proposition} \label{prop:unitarity_bound}
Let $\E$ be a unital and trace-preserving quantum operation. 
Then $D := D(\E)$, $r := r(\E)$ and $u := u (\E)$ are related via
\begin{equation}
c_d \sqrt{ u  + \frac{2 d r}{d-1}-1 } 
\leq D
\leq d^2 c_d \sqrt{u  + \frac{2dr}{d-1} -1 }, \label{eqn:unitarity_estimate}
\end{equation}
where $c_d = \tfrac{1}{2}\left(1-\tfrac{1}{d^2}\right)^{1/2} \in \left[ \frac{\sqrt{3}}{4}, \frac{1}{2} \right]$ that only depends on the system dimension $d$.
\end{proposition}

To deduce \autoref{thm:unitarity_witness} from this statement, let us start 
 with assuming that \eqref{eqn:unitarity_scaling} holds. 
Inserting this expression for $u$ into the upper bound provided by \autoref{prop:unitarity_bound} yields
\begin{equation*}
 D
\leq  d^2 c_d \sqrt{ \left( 1 - \frac{d r}{d-1}\right)^2 + \mathcal{O}(r^2) 
+ \frac{2d r}{d-1} - 1} 
= d^2 c_d \sqrt{ \frac{d^2}{(d-1)^2} r^2 + \mathcal{O} \left( r^2 \right)}
= \mathcal{O} \left( r \right),
\end{equation*}
as claimed. 
Conversely, suppose by contradiction that $u = \left( 1 - \frac{d r }{d-1} \right)^2 + \mathcal{O}(r)$. 
Employing the lower bound provided by \autoref{prop:unitarity_bound} in a similar fashion assures $D (\E) = \mathcal{O} (\sqrt{r} )$ which definitely does not scale linearly in $r$.

In order to establish the remaining statements -- \autoref{prop:unitarity_bound} and  \autoref{prop:unitarity_vs_rate} -- it is very useful to choose a particular Liouville representation of 
error channels $\E$. Concretely, we let $\left\{B_1,\ldots,B_{d^2} \right\}$
be a unitary operator basis obeying $B_1 = \frac{1}{\sqrt{d}}I$ and $\tr \left( B_i^\dagger B_j \right) = \delta_{i,j}$ (e.g.\ the normalized Pauli's with the identity as first element). 
If defined with respect to such a basis, $L(\E)$ admits the following block structure
\begin{equation}
L(\E) =
\left(
\begin{array}{cc}
\frac{1}{d} \tr \left( \E (I) \right) & e_{\mathrm{sdl}} \\
e_{\mathrm{nu}} & E_\mathrm{u}
\end{array}
\right), \label{eqn:block_structure}
\end{equation}
where $e_{\mathrm{sdl}},e_{\mathrm{nu}} \in \mathbb{C}^{d^2-1}$ encapsulate state dependent leakage and non-unitarity, respectively. 
With such a Liouville representation, the unitarity of $\E$ is proportional to the squared Frobenius  (or Hilbert-Schmidt) norm of the unital block $E_u$~\cite{Wallman2015}[Proposition 1]:
\begin{equation}
u (\E) = \frac{1}{d^2-1} \left\| E_{\mathrm u} \right\|_2^2 .
\label{eqn:unitarity_2norm}
\end{equation} 
Moreover, such a block-matrix structure lets us establish the following relation~\cite{Wallman2015}[Proposition 9]
\begin{equation}
\| J (\E) \|_2^2 = (d^2+1) u (\E) + \| e_{\mathrm{nu}} \|_{\ell_2}^2 + \| e_{\mathrm{sdl}} \|_{\ell_2}^2 + \frac{1}{d} \tr \left( \E (I) \right). \label{eqn:unitarity_choi}
\end{equation}
between the unitarity and the channel's associated Choi matrix.
Having laid out these relations, we are ready to prove the main technical result of this section. 

\begin{proof}[Proof of \autoref{prop:unitarity_bound}]
We start with pointing out that the statement's assumptions assure that both $e_{\mathrm{nu}}$
and $e_{\mathrm{sdl}}$ vanish. This considerably simplifies the block structure \eqref{eqn:block_structure} of $L( \E)$ as well as relation \eqref{eqn:unitarity_choi}.
At the heart of this statement is an inequality that relates the diamond norm of any map $\mathcal{M}$ to different Schatten-norms of its corresponding Choi matrix:
\begin{equation}
\frac{1}{d}\| J (\mathcal{M}) \|_1 \leq \| \mathcal{M} \|_\diamond \leq  \| J(\mathcal{M} ) \|_1,
\label{eq:diamond_relation}
\end{equation}
see e.g.~\cite{Wallman2014}[Lemma 7].
Recalling $D(\E) = \frac{1}{2} \| \Delta \|_\diamond$ and weakening this estimate by employing the Schatten norm inequalities $\| X \|_2 \leq \| X \|_1 \leq \mathrm{rank}(X) \| X \|_2$ allows us to deduce
\begin{equation}
\frac{1}{2d} \| J(\Delta) \|_2 \leq D (\E) \leq \frac{d}{2} \| J(\Delta) \|_2,
\label{eqn:unitarity_bound_aux0}
\end{equation}
because $J(\Delta)$ has at most rank $d^2$.
Note that an analogous relation can be derived using the diamond norm bound presented in~\cite{Kliesch2015} instead of \eqref{eq:diamond_relation}.
As a matter of fact, the assumptions on $\E$ allow us to calculate 
$\| J(\Delta) \|_2$ explicitly.
To do so, start with
\begin{equation}
\| J (\Delta) \|_2^2 = \| J (\I - \E) \|_2^2
= \| d | \psi_{\mathrm{Bell}} \rangle \langle \psi_{\mathrm{Bell}} | - J(\E) \|_2^2 
= d^2 \langle \psi_{\mathrm{Bell}} , \psi_{\mathrm{Bell}} \rangle^2
- 2 d \langle \psi_{\mathrm{Bell}} | J (\E) | \psi_{\mathrm{Bell}} \rangle + \| J (\E) \|_2^2
\label{eqn:unitarity_bound_aux2}
\end{equation}
and note that the second term is related to the average error rate via 
\begin{equation*}
\langle \psi_{\mathrm{Bell}}| J (\E) | \psi_{\mathrm{Bell}} \rangle 
= (d+1) F_{\mathrm{avg}}(\E) -1 = (d+1)(1- r (\E)) -1.
\end{equation*}
This can readily be deduced from \eqref{eqn:infid_formula} by inserting the identity 
$\tr \left( L (\E) \right) = d \langle \psi_{\mathrm{Bell}} | J (\E) | \psi_{\mathrm{Bell}} \rangle$ and
noting that
$\tr \left( \E (I) \right) = \tr (I) = d$ holds, because $\mathcal{E}$ is trace-preserving. 
In turn, equation \eqref{eqn:unitarity_choi} allows to replace the last term in Eq.~\eqref{eqn:unitarity_bound_aux2} by
\begin{equation*}
\| J (\E) \|_2^2 = (d^2-1) u(\E) + \frac{1}{d^2} \tr \left( \E (I) \right)^2+\| e_{\mathrm{sdl}} \|_{\ell_2}^2 + \| e_{\mathrm{n}} \|_{\ell_2}^2 
= (d^2-1)u(\E) + 1,
\end{equation*}
where we have used our assumptions that $\E$ is both unital and trace preserving to considerably simplify this expression.
Inserting these identities into Eq.~\eqref{eqn:unitarity_bound_aux2} reveals
\begin{align*}
\| J (\Delta) \|_2^2
=& d^2 - 2d(d+1)(1- r(\E)) +2d+ (d^2-1) u(\E) + 1 \\
=& (d^2-1)u(\E) + 2 d (d+1) r (\E) + - d^2 + 1 \\
=& (d^2 -1) \left( u(\E) + \frac{2 d r (\E)}{d-1} - 1 \right).
\end{align*}
Plugging this explicit expression into the inequality chain \autoref{eqn:unitarity_bound_aux0}
then establishes the claim.
\end{proof}

Finally, we provide a proof of Propostion \ref{prop:unitarity_vs_rate}.

\begin{proof}[Proof of \autoref{prop:unitarity_vs_rate}]
The claim can be deduced from the fundamental norm inequality $\| X \|_1^2 \leq \mathrm{rank} (X) \| X \|_2^2$.
Now, let $L(\E)$ be the particular block matrix representation \eqref{eqn:block_structure}.
By construction $E_\mathrm u$ has rank at most $(d^2-1)$ and we infer that
\begin{equation}
\tr ( E_{\mathrm{u}} )^2 
 \leq \| E_{\mathrm u} \|_1^2 \leq \mathrm{rank} (E_\mathrm u) \| E_\mathrm u \|_2^2
= (d^2-1)^2 u (\E)
\label{eqn:unitarity_vs_rate_aux1}
\end{equation}
must hold, where we have employed Eq.~\eqref{eqn:unitarity_2norm}. 
Also, Formula~\eqref{eqn:infid_formula} together with the definition of the error rate implies
\begin{equation*}
\tr ( L (\E) ) + \tr ( \E (I) ) 
= d(d+1)F_{\mathrm{avg}}(\E)
=d(d+1)(1- r ( \E ) ).
\end{equation*}
This in turn allows us to calculate
\begin{align*}
\tr ( E_\mathrm{u} )
=& \tr ( L (\E) ) - \frac{1}{d} \tr ( \E (I) ) 
= \tr ( L (\E) ) + \tr ( \E (I) ) - \frac{d+1}{d} \tr ( \E (I) ) \\
=& d(d+1)(1- r (\E)) - \frac{d+1}{d} \tr ( \E (I) ) 
\geq  d(d+1)(1- r(\E) ) - (d+1) \\
=& d(d+1) \biggl( \frac{d-1}{d} - r (\E) \biggr),
\end{align*}
and combining this estimate with \eqref{eqn:unitarity_vs_rate_aux1} readily yields the claimed bound.
\end{proof}

\end{document}